\documentclass[12pt]{scrartcl}

\usepackage[english]{babel}

\usepackage[a4paper,top=2cm,bottom=3cm,left=3cm,right=3cm,marginparwidth=1.75cm]{geometry}

\usepackage{amsmath}
\usepackage{graphicx}
\usepackage{subcaption}
\usepackage[colorlinks=true, allcolors=blue]{hyperref}
\usepackage{amsthm}
\usepackage{amssymb}
\usepackage{natbib}
\usepackage{booktabs}
\usepackage{bbm}

\theoremstyle{plain}

\newtheorem{proposition}{\sffamily \bfseries Proposition}[section]
\theoremstyle{definition}
\newtheorem{definition}{\noindent \sffamily \bfseries Definition}[section]
\theoremstyle{remark}
\newtheorem{notation}{\noindent \sffamily \bfseries Notation}[section]

\newcommand{\sectionref}[1]{Section \ref{#1}}
\newcommand{\subsectionref}[1]{Subsection \ref{#1}}
\newcommand{\figureref}[1]{Figure \ref{#1}}
\newcommand{\tableref}[1]{Table \ref{#1}}
\newcommand{\definitionref}[1]{Definition \ref{#1}}
\newcommand{\appendixref}[1]{Appendix \ref{#1}}

\newcommand{\propositionref}[1]{Proposition \ref{#1}}

\newcommand{\Ef}{\mathbb{E}}                            
\newcommand{\Rf}{\mathbb{R}}                            
\newcommand{\Pf}{\mathbb{P}}                            
\newcommand{\Var}{\mathbb{V}ar}                         

\title{Deep Hedging of Green PPAs in Electricity Markets}
\author{Biegler-König, Richard\thanks{Steag Iqony Group, richard.biegler-koenig$@$iqony.energy} ~ and Oeltz, Daniel\thanks{Fraunhofer SCAI, daniel.oeltz$@$scai.fraunhofer.de}}

\begin{document}
\maketitle

\begin{abstract}
In power markets, Green Power Purchase Agreements have become an important contractual tool of the energy transition from fossil fuels to renewable sources such as wind or solar radiation. Trading Green PPAs exposes agents to price risks and weather risks. Also, developed electricity markets feature the so-called \emph{cannibalisation effect}: large infeeds induce low prices and vice versa. As weather is a non-tradable entity the question arises how to hedge and risk-manage in this highly incomplete setting. We propose a "deep hedging" framework utilising machine learning methods to construct hedging strategies. The resulting strategies outperform static and dynamic benchmark strategies with respect to different risk measures. 
\end{abstract}

\section{Introduction}
\label{section_introduction}
    The transition of the electricity production from fossil fuels to renewable sources is one of the most important tasks of the present and future. State financed feed-in tariffs have helped to kick-off this transition. Today though, the production from renewable sources is more and more brought to liberalised electricity markets using so-called \emph{Green Power Purchase Agreements} (PPA).
    
    PPAs are contracts that sell the production of a generating asset at a fixed price. PPAs are not new, they have helped to secure the large investment sums needed when building conventional power plants for decades. This is achieved by guaranteeing a secure stream of income to the owner of the asset. At the same time the buyer is exposed to market risk. A PPA is called "Green" when its underlying generating asset uses a renewable source.
    
    Due to the fact that electricity is not storable its production has to match consumption at every point in time. This results in huge variations in prices. Some of those variations are driven by large differences in demand: high demand during the week, low demand on the weekend; high demand in winter (heating) and summer (cooling), low demand in spring. Also, (non-) availability of plants can lead to huge price jumps. Finally, in markets with a lot of installed green energy (in particular windmills or solar farms), the weather has a huge impact on prices. This impact is called \emph{cannibalisation effect} as it results in very low prices when it is sunny and windy whereas it leads to high prices when solar radiation or wind are absent.
    
    Trading a Green PPA exposes us to these price fluctuations, but, as seen above, our risk profile also depends on the amount of electricity produced by the underlying asset. Electricity markets allow, to a certain degree, to hedge against price risks as they feature liquidity in swap-like derivatives delivering power at a fixed price over future periods. There is no market to trade the quantity risk though, that is, the weather. Thus, we find ourselves in a highly incomplete market and all the canonical strategies, such as delta hedging, are just one out of a continuum of possible strategies. Even more, it is hard to exactly and completely describe the interplay between weather and power prices.
    
    Finding and implementing a good hedging strategy is of the utmost importance for individual electricity traders but also for the success of the aforementioned energy transition. Our research idea is to apply machine learning methods in combination with some theory from pricing on incomplete markets to overcome this problem. 
    
    In this paper, we set up a stochastic model that captures some of the most important properties of the two risk factors of Green PPAs, weather-induced quantity as well as electricity prices. We then set up a machine learning framework to optimise an agent's individual hedging strategy.
    
    We show that the resulting hedge strategies outperform static and dynamic benchmark strategies. We apply some intuition as well as specific \emph{domain knowledge} about energy markets to \emph{interpret} the applied algorithms (see \citet{garcke_2023} for these classifications).
    
    The structure of the paper is as follows: in \sectionref{section_markets} we will give a very brief summary of the design of electricity markets (we will exemplify in this paper with the German market in mind). Also, we will discuss the main contractual features of Green PPAs citing more literature for further reference. Then, in \sectionref{section_model}, we will motivate and introduce the model framework. \sectionref{section_deepHedging} we will introduce the "deep hedging" framework and how we adjust it to our use case. Finally, in \sectionref{section_experiments}, we will present and discuss empirical results. We conclude with a summary of future research.
    
    Our main sources are: For the structure and class of our proposed model \citet{wagner_2019} and \citet{wagner_2012}. We rely on \citet{buehler_2019} for the idea and general setup of "deep hedging".
    
    
\section{Electricity Markets and Green PPA contracts}
\label{section_markets}

    \paragraph{Electricity markets:} In countries with liberalised electricity markets, power is a traded commodity. The main, outstanding feature of electrical power is that it is not storable. This leads to a multi-layered market design, in which physical restrictions and regularisation increase the closer one gets to the time of delivery. Usually, there is a market trading futures/forwards on constant deliveries of power at a fixed price for years, quarters and months. These then settle against the so-called spot, usually an auction for the individual hours of the following day. There are intraday auctions with a granularity of 15 minutes and also continuous intraday trading. Finally, the grid operator enforces that supply matches demand in various reserve markets.
    
    \paragraph{Green PPAs} A detailed discussion of the various contractual features of Green PPAs is given in \citet[Section 1.2]{biegler_2022}. As mentioned above, PPAs deliver the electricity of a specified asset for a fixed price. In case this asset has a renewable "fuel" we speak of a Green PPA. As such, they are generally linear in both quantity and price, but there are exceptions, in particular when it comes to the treatment of periods with a negative price (we refer to \citet{biegler_2015} for more information on the extraordinary structure of power prices).
    
    One important aspect of Green PPAs is their delivery period. Due to the layered structure of power markets, one can divide this period into three parts that require different techniques when it comes to valuation, hedging and risk-managing. The long-term part is beyond the time point for which futures can be traded (usually three years into the future). Then follows the mid-term for which liquidly traded futures are available and can be used for hedging purposes. This ends with the day-ahead spot auction. With all futures having settled, the intraday period then starts. In our paper, we will look at a simplified setup which features traded futures that settle against a continuous spot.
    
    There is still little research available when it comes to the subject of valuation and hedging of Green PPAs. The authors of \citet{hirth_2013} and \citet{tranberg_2020} conduct comprehensive studies about the cannibalisation effect. \citet{wagner_2012} proposes a model that takes the interplay between weather and power prices into consideration and we will build on his ideas. With a spot model, the authors of \citet{biegler_2022} set up a valuation framework and calculate fair values as well as delta hedge quantities for Green PPAs. They also show that the afore mentioned special treatment of negative hours leads to a non-linear payoff structure and transforms the PPA into an option-style product. Their paper concludes with the comment that, as one faces highly incomplete markets (we refer to \citet[Chapter 4.3]{kiesel_2004} for the role of incompleteness in pricing and hedging), their "risk-neutral" approach is but one of a continuum of possible pricing approaches - and this is exactly the start-off point of our research here.
    
    Finally, we will look at the German electricity market as our main example. The reason for this is twofold: on the one hand, the German market is quite mature with all of the above market-layers firmly established. On the other hand, Germany has seen a large increase in installed capacity of renewable assets which has two main consequences: firstly, Green PPAs are highly relevant. And secondly, wind speed and solar radiation are strongly linked to power prices by the afore mentioned \emph{cannibalisation effect}. The latter poses interesting modelling and hedging challenges that we will further explore in the following.

\section{Model}
\label{section_model}
    In order to optimise a hedging strategy for a PPA contract we will now propose a stochastic model for the evolution of renewable infeeds as well as power forward prices. The resulting forward prices will be the building blocks of the hedging strategy.
    
    We begin by introducing stochastic processes for the infeeds. We will make use of the ideas proposed in \citet{wagner_2012} which ensure that infeeds of the model do not exceed the installed capacity or become negative. This is achieved by modelling the efficiency of the infeeds (as a percentage) rather than the infeeds directly. We also ensure that the model produces infeed paths that are, in expectation, unbiased with respect to our given forecasts.
    
    Secondly, we will set up a process for power prices that depends on the infeeds and thus shows the \emph{cannibalisation effect}. In this, we will follow the ideas of \citet{wagner_2019}. There, the authors propose a HJM-framework (Heath-Jarrow-Morton, see \citet{heath_1992}) for consistent power prices across time. Our power prices will be constructed from three components. Firstly, they will be arbitrage-free with respect to a given forward curve. Secondly, they will exhibit idiosyncratic stochastic risk. Thirdly, they will dependent on the above infeeds (this latter component is called the \emph{structural component} in \citet{wagner_2019}).
    
    Last but not least, in order to approximate better the typical flow of information, we will present a slight modification to our model for which the coupling between infeeds and prices is only activated at discrete time points. By this, we take the availability of new weather forecast into consideration.

    \subsection{The infeed of renewable energies}
    \label{subsection_quantities}
        We begin by defining processes for the efficiencies of the infeed of renewable sources. We will assume there are $n$ such sources. Typically $n$ will be two or three, describing wind (possibly on- and offshore) as well as solar radiation. By efficiency we mean the fraction of the total installed capacity of that technology producing electric power at a certain time.
    
        \begin{notation} [Efficiency of Renewables]
        \label{not_quantities}
            We denote by $Q_i(T)$ the efficiency of the i'th technology in time $T$ as seen from time $t=0$. We have $Q_i(T) \in (0,1)$. Similarly, let $Q_i(t,T)$ denote this efficiency but seen from some time $t$, $0 \leq t \leq T$.
        \end{notation}

        Clearly, we have the relation 
        \begin{align}
        \label{eq:ForwardRelation}
            Q_i(t,T) = \Ef^\mathbb{P}[Q_i(T) | \mathcal{F}_t]
        \end{align} where $\mathbb{P}$ is the real-world measure. In the following, we will understand the process $Q(T)$ to be vector-valued when omitting the technology index $i$. Next, we will propose a concrete process for the efficiencies that we will use in the following.

        \begin{definition} [Model for the Efficiency]
        \label{def_quantities}
            We define the efficiency process for technology $i$ and time $T \geq 0$ as:
            \begin{align*}
                Q_i(T) = \varsigma(X_i(T) + \phi_i(T))
            \end{align*}
            where $\varsigma(\cdot)$ is a sigmoid invertible function such as the logistic or smoothstep functions. Furthermore, $X_i(T)$ is an Ornstein-Uhlenbeck process given by the dynamics:
            \begin{align*}
                dX_i(t) = -\kappa_i X_i(t) dt + \sigma_i dW_i(t)
            \end{align*}
            Here, $\kappa$ and $\sigma$ are the speed of mean-reversion and the volatility and $W_i(t)$ is a Brownian motion.
            The deterministic function $\phi_i(T)$ has to be chosen so that the initial condition 
            \begin{align*}
                Q_i(0,T) = \Ef\left[\varsigma\left(X_i(T) +\phi_i(T) \right) | \mathcal{F}_0 \right]
            \end{align*}
            is fulfilled. This condition ensures that the model infeeds are unbiased with respect to the initial forecast $Q_i(0,T)$. In other words, we "believe" in our forecast.
        \end{definition}
        \begin{figure}
            \centering
            \includegraphics[width=0.49 \textwidth]{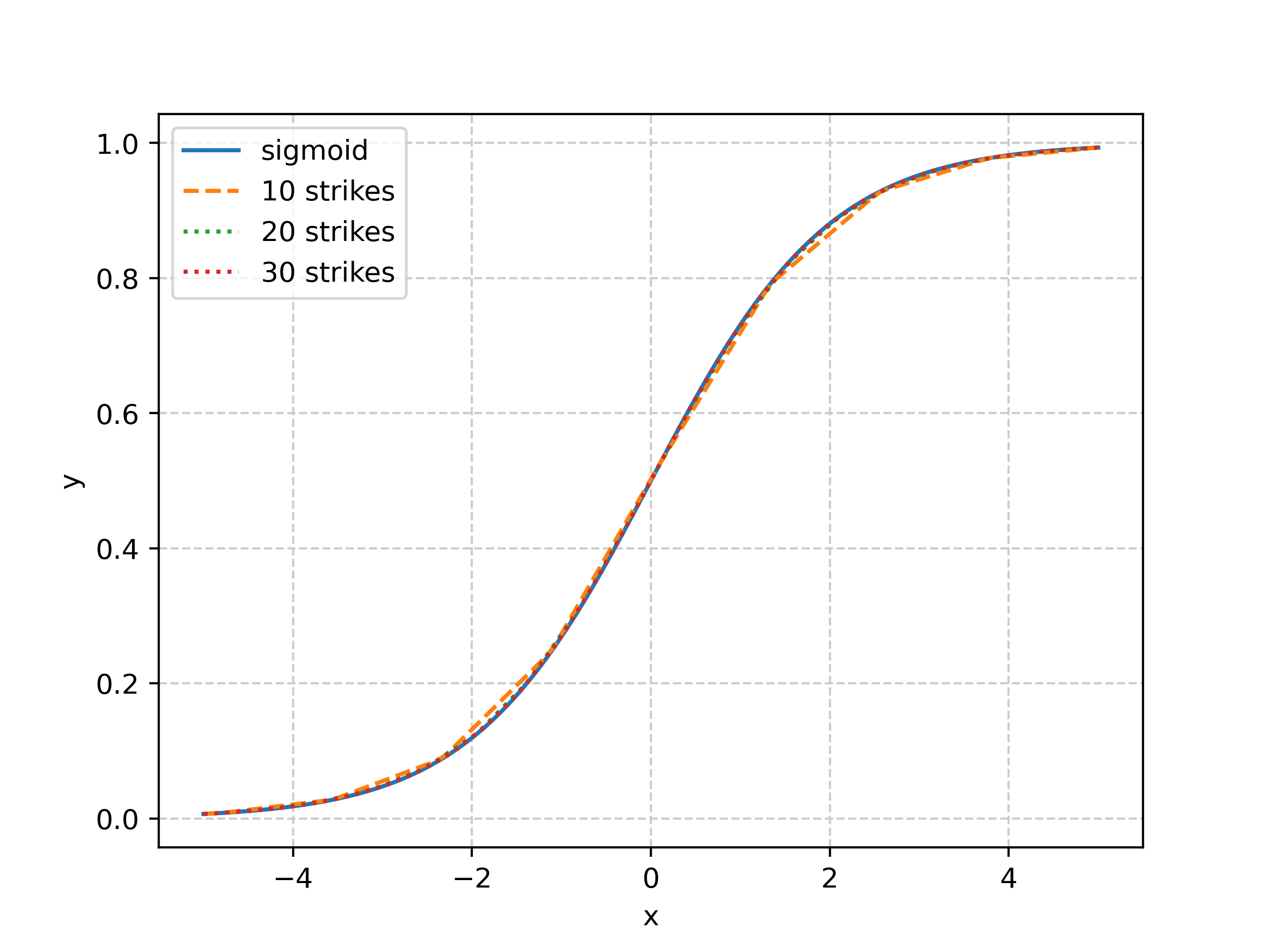}
            \includegraphics[width=0.49\textwidth]{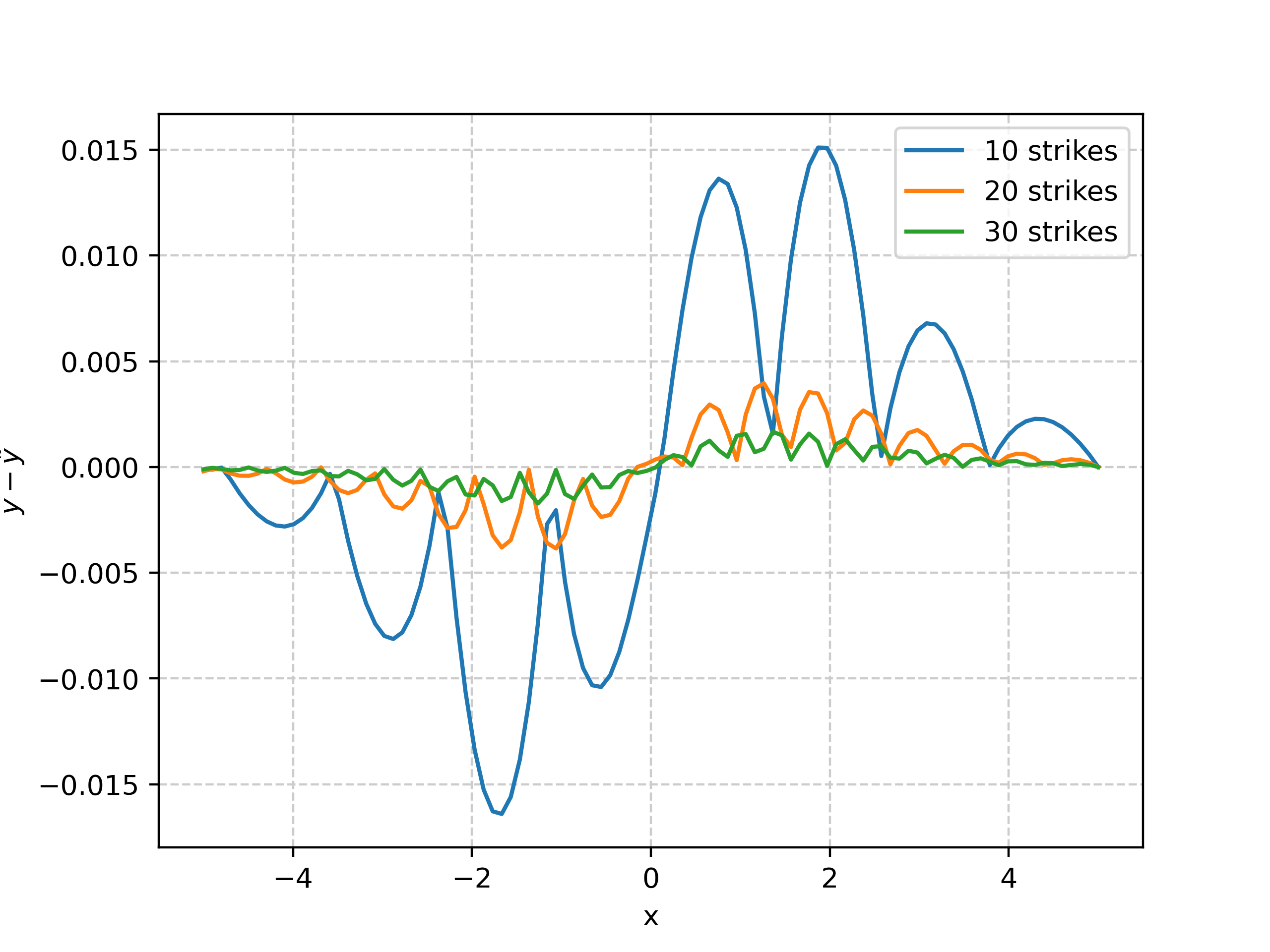}
            \caption{Interpolants and point-wise error approximating the sigmoid function with 10, 20 and 30 points on a uniform grid between -5 and 5.}
            \label{fig:sigmoid_approx}
        \end{figure}
        Since there is no analytic formula to evaluate equation (\ref{eq:ForwardRelation}) (see \citet[p. 7]{wagner_2012} for details) we resort to the following trick: we approximate the sigmoid function by a sum of functions of the form $c_j(x) = \max(0, x - k_j)$, i.e. for $k_j$, $1\leq j\leq N$, we determine weights $w_j$ such that
        \begin{align*}
            \varsigma(x) \approx \sum_{j=1}^N w_j \max(0, x - k_j)
        \end{align*}
        and the sum is just the linear interpolation on the grid $k_j$. To give an impression of the quality of fit we show in \figureref{fig:sigmoid_approx} an example using a uniform grid between -5 and 5 and the resulting linear interpolants and  point-wise errors using 10, 20 and 30 points. 
        
        As to the calculation of the above approximation, we remark that these functions have exactly the form of the payoff of a vanilla call option. With this in mind and the linearity of the expectation we can evaluate the expression (\ref{eq:ForwardRelation}) as a sum of call prices using the following proposition:
        
        \begin{proposition} [Call option on an Ornstein-Uhlenbeck process]
        \label{prop_approx_call}
            The price of a call option $C(t,T)$ on an underlying modelled by an Ornstein-Uhlenbeck process $X(T)$ with time-dependent level $\mu(t)$, constant speed of mean-reversion $\kappa$, constant volatility $\sigma$, strike $K$ is given by:
            \begin{align*}
                C(t,T) &= (g(t,T) - K) ~ \Phi(d) + \bar{\sigma} ~ \phi(d)
            \end{align*}
            where the helper functions are given as:
            \begin{align*}
            	d &= \frac{-\left(X_t e^{-\kappa(T-t)} + \int^T_t e^{-\kappa(T-u)} \kappa \mu(u) du  - K\right)}{\bar{\sigma}(t,T)} \\
                g(t,T) &= X_t e^{-\kappa(T-t)} + \int^T_t e^{-\kappa(T-u)} \kappa \mu(u) du \\
            	\bar{\sigma}(t,T) &= \sigma^2 \frac{1}{2 \kappa} \left(1 - e^{-2\kappa(T-t)}\right)
            \end{align*}
            and $\Phi(\cdot)$ and $\phi(\cdot)$ denote the standard normal distribution and density functions.
        \end{proposition}
        \begin{proof}
            The proof is given in \appendixref{section_appendix_ouCall}.
        \end{proof}

    \subsection{Power Prices}
    \label{subsection_prices}
        As mentioned before, we construct power prices along the lines of the model framework proposed in \citet{wagner_2019}. Clearly, we want our price model to be free of arbitrage with respect to the observed forward curve, which, in the case of the underlying power is called the \emph{HPFC} (for more information about the special properties and construction of the \emph{hourly-price-forward-curve} we refer to \citet{biegler_2015}).

        \begin{notation} [HPFC]
        \label{not_hpfc}
            We denote by $f(T)$ the power forward curve as seen from $t=0$. As before, we denote by $f(t,T)$ the forward curve as seen from some time $t \geq 0$, that is: $f(t,T) = \Ef^\mathbb{Q}[f(T) | \mathcal{F}_t]$.
        \end{notation}

        Next, we propose definitions for the structural component and the idiosyncratic risk component of the power price.

        \begin{definition} [Idiosyncratic Component]
        \label{def_idiosyncratic}
            We model the idiosyncratic risk of the power price using process $X^P(T)$. We define this stochastic process to be another Ornstein-Uhlenbeck process given by the dynamics:
            \begin{align*}
                dX^P(t) = -\kappa^P dt + \sigma^P dW^P(t)
            \end{align*}
            for speed of mean-reversion $\kappa^P$, volatility $\sigma^P$ and Brownian motion $W^P(t)$. As usual, we also define a conditional version and denote it by $X^P(t,T)$.
        \end{definition}
        
        \begin{definition} [Structural Component]
        \label{def_structural}
            For the structural component we define a function $g(Q(T))$ linking infeeds and prices. In this paper we make use of a so-called \emph{supply-curve} $S$ and define:
            \begin{align*}
                g(Q(T)) = S(1 - w^TQ(T))
            \end{align*}
            where $w$ is a vector of weights describing the ratio of the installed capacity of the different sources of renewable energy when compared to the total installed capacity. The construction ensures that the argument of function $S$ is positive and smaller than one.
        \end{definition}
        
        The intuition behind the supply curve is as follows: we have chosen as weights the fraction of the installed capacity of the corresponding renewable technology. The difference between $1$ and the sum of those weighted efficiencies then gives us the amount of electricity that has to be produced by conventional plants. And those plants enter the spot-auction at their (non-zero) marginal costs and thus are the main driver of prices. For the sake of computational simplicity we will use the identity function as the supply curve in the remainder of this paper. 

        We will now combine the three components as illustrated in \citet{wagner_2019}.
        
        \begin{definition} [Model for the Power Price]
        \label{def_powerPrice}
            We define the power forward price in time $t$ for future time $T$ in geometric fashion as:
            \begin{align*}
                f(t,T) = f(T) ~ \Ef\left[\frac{X^P(T)}{\Ef[X^P(T)]} | \mathcal{F}_t \right] ~ \Ef\left[\frac{g(Q(T)}{\Ef[g(Q(T))]} | \mathcal{F}_t\right]
            \end{align*}
        \end{definition}
        
        The no-arbitrage property of \definitionref{def_powerPrice} is easily checked by setting $t=0$.
        
    \subsection{Time grid of weather forecasts}
    \label{subsection_timeFrame}
        Usually, a number of major weather forecasts is calculated and published daily mostly by state-owned research facilities. Minor updates are being extrapolated from them by commercial providers who also usually transform the weather variables into infeeds. To account for the above mentioned main runs of the weather forecasts we assume they arrive at discrete time points $t_i$. We then introduce a new time grid as:
        \begin{align*}
            t_{-} := \max_{i, t_i<t} \{ t_i \}
        \end{align*}
        This leads to the adjusted model for the power price as defined before as follows:
        
        \begin{definition} [Model for the Power Price with arrival of forecasts]
        \label{def_powerPrice_forecasts}
            We define the power forward price with adjustment to the arrival of new weather forecasts as:
            \begin{align*}
                f(t,T) = f(T) ~ \Ef\left[\frac{X^P(T)}{\Ef[X^P(T)]} | \mathcal{F}_t \right] ~ \Ef\left[\frac{g(Q(T)}{\Ef[g(Q(T))]} | \mathcal{F}_{t_{-}}\right]
            \end{align*}
            This is a slightly modified version of \definitionref{def_powerPrice} with conditioning of the structural component to the arrival times of forecasts.
        \end{definition}
    
    \begin{figure}
        \centering
        \includegraphics[width=0.99 \textwidth]{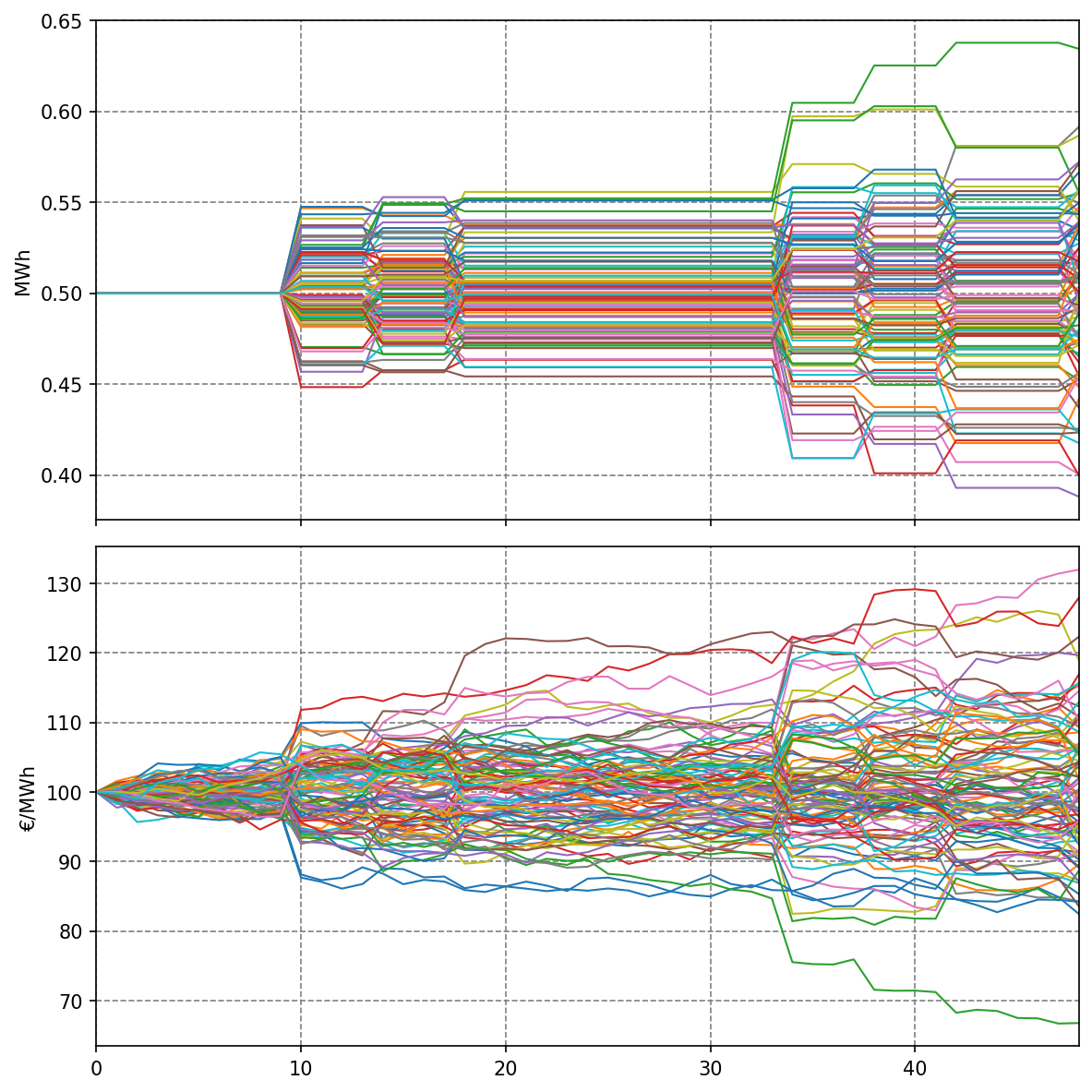}
        \caption{Paths of the stochastic model introduced in \sectionref{section_model}, infeed above (stylised wind onshore), power price below. Parameters used were: $t=0, T=48$, forecast $Q_1(0, T) = 0.5$, forward price $f(0, T) = 100$. Arrival of major weather forecasts at $t \in \{10, 14, 18, 34, 38, 42\}$. We have set $\varsigma(x) = \frac{1}{1+e^{-t}}$. Parameters of the stochastic processes are: $\kappa_1 = \kappa_2 = 0.1$, $\sigma_1 = \sigma_2 = 3.0$, $\kappa^P = 0.5$, $\sigma^P = 0.8$. The technology weights are $w_1 = 0.8, w_2 = 0.2$.}
        \label{fig:paths}
    \end{figure}
    
    \figureref{fig:paths} illustrates sample paths of the model. One can clearly see the influence of the infeeds on the price at the times of arrival of new weather forecasts. Furthermore, one sees the time-continuous idiosyncratic risk component of the prices.

\section{Deep Hedging}
\label{section_deepHedging}
    
    Forwards and futures traded on electricity provide the toolbox to hedge against price risks. The second source of risk in a PPA is the uncertain quantity that depends on the weather. This risk component is not a traded commodity on the financial market and thus cannot be replicated. Thus, we are facing a highly incomplete market. We are referring to \citet[Chapter 4.3]{kiesel_2004} for a discussion of incompleteness and the \emph{Fundamental Theorem of Asset Pricing}.
    
    In an incomplete market, we will need some criterion to chose from the continuum of possible hedging strategies and then identify their corresponding fair values. This criterion must and will depend on the trader's individual risk appetite and thus the resulting price will be called an \emph{indifference price}. There are two main types of this criterion in the literature: firstly, making use of utility functions and secondly, applying risk measures. We will briefly describe both of these types but our experiments will make use of risk measures only. In introducing the concepts and notation we will follow \citet[Section 3]{buehler_2019}. We set up a time horizon $T$ and discrete trading dates $0 = t_0 < t_1 < \ldots < t_n = T$. We denote the derivative at hand by the real random number $Z$. The portfolio of hedge instruments will be denoted by $S(t)\in\mathbb{R}^{n_S}$. $X_t$ is a $n_X$-dimensional random variable describing the state of the system at time $t$. A hedging strategy is a function $\delta(.):\mathbb{R}^{n_X}\mapsto \mathbb{R}^{n_S}$ and  $\mathcal{H}$ is the set of all valid hedging strategies. Note that for a complete market this hedging strategy is equal to the well-known $\Delta$-hedge (cf. \citet{biegler_2020} for hedging strategies in energy markets). Furthermore, we will assume trading is free of charge.
    
    \paragraph{Utility function} Let $U(Y)$ be a utility function of a position $Y$. We set up the optimisation problem to find an optimal hedging strategy for $Z$ using a vector of hedge instruments $S$:
    
    \begin{align*}
        \sup_{\delta \in \mathcal{H}} \Ef\left[U\left(Z + \sum^{n-1}_{k=0} \delta(X_k)\cdot (S_{k+1} - S_k)\right)\right]
    \end{align*}
    
    With this, the \emph{indifference price} of derivative $Z$ under the utility function $U(Y)$ will be denoted by $p^U_0$. It will be exactly the amount of money that, together with derivative $Z$ and the hedge, gives the same utility to the investor as not investing in $Z$ at all. This is:
    
    \begin{align*}
        \sup_{\delta \in \mathcal{H}} \Ef\left[U\left(p^U_0\! + Z\! + \sum^{n-1}_{k=0} \delta(X_k) (S_{k+1}\! -\! S_k)\right)\right] = \sup_{\delta \in \mathcal{H}} \Ef\left[U\left(\sum^{n-1}_{k=0} \delta(X_k) (S_{k+1}\! -\! S_k)\right)\right]
    \end{align*}
    
    For example, using the exponential utility so that $U(x) = -e^{-\lambda x}$ for some $x \in \Rf$ (c.f. \cite{henderson_2004}) allows to further simplify the above expression and write:
    \begin{definition}[Indifference price (utility function)]
        For a derivative $Z$ we define the \emph{indifference price} under the exponential utility function as:
        \begin{align*}
            p^U_0 &= \frac{1}{\lambda} \ln\left(
            \frac{\sup_{\delta \in \mathcal{H}} \Ef\left[\exp\left(-\lambda \left(Z + \sum^{n-1}_{k=0} \delta(X_k)\cdot (S_{k+1} - S_k)\right)\right)\right]}
            {\sup_{\delta \in \mathcal{H}} \Ef\left[\exp\left(- \lambda \left(\sum^{n-1}_{k=0} \delta(X_k)\cdot (S_{k+1} - S_k)\right)\right)\right]}\right)
        \end{align*}
    \end{definition}
    
    \paragraph{Risk measure} Now, let $\rho(Y)$ be a convex risk measure applied to a position $Y$. We formulate finding the hedge strategy as an optimisation problem:
    
    \begin{align*}
        \inf_{\delta \in \mathcal{H}} \rho\left(Z + \sum^{n-1}_{k=0} \delta(X_k)\cdot (S_{k+1} - S_k)\right)
    \end{align*}
    
    With this, the \emph{indifference price} of derivative $Z$ under the risk measure $\rho(Y)$ will be denoted by $p^\rho_0$. It will be exactly the amount of money that will make the investor indifferent between investing in $Z$ and the hedge as compared to hedging a zero position:
    
    \begin{align*}
        \rho(0) &= \inf_{\delta \in \mathcal{H}} \rho\left(p^\rho_0 + Z + \sum^{n-1}_{k=0} \delta(X_k)\cdot (S_{k+1} - S_k)\right)
    \end{align*}
    
    As our model does not enable statistical arbitrage (see the discussion in \subsectionref{subsection_parameters}) the not-investing strategy has a zero value $\rho(0) = 0$. Making use of the cash-invariance of $\rho(Y)$ this leads to:
    
    \begin{definition}[Indifference price (risk measure)]
        For a derivative $Z$ we define the \emph{indifference price} under the risk measure $\rho(\cdot)$ as:
        \begin{align*}
            p^\rho_0 &= \inf_{\delta \in \mathcal{H}} \rho\left(Z + \sum^{n-1}_{k=0} \delta(X_k)\cdot (S_{k+1} - S_k)\right)
        \end{align*}
    \end{definition}
    
    As we have seen so far, pricing and hedging with respect to a fixed risk measure or utility function can be formulated as an optimization problem over some space $\mathcal{H}$ of all valid hedging strategies. To solve this problem numerically, we have to come up with a certain subset $\mathcal{H}_\Theta \subset \mathcal{H}$ whose strategies $\delta_\Theta $ depend on a finite number of parameters $\Theta \in\mathbb{R}^N$ for some $N>0$. In the following we use a neural network where the parameter $\Theta$ denotes the network weights and biases as parameterisation for possible strategies, which ends up in the approach that has been introduced in \citet{buehler_2019} under the name of Deep Hedging. All of the above carries over to this finite dimensional case, e.g. we may find the optimal strategy with respect to a utility function as the strategy with optimal weights $\Theta^\star$ by 
    \[
    \inf_{\theta\in\mathbb{R^N}}\rho\left(Z + \sum^{n-1}_{k=0} \delta_{\Theta}(X_k)\cdot (S_{k+1} - S_k)\right)
    \]
    
    Our numerical experiments will be conducted applying the approach using risk measures. We choose the \emph{Expected Shortfall} (ES): 
    \begin{definition}[Expected Shortfall]
    \label{def_ES}
        The Expected Shortfall of a position $Y$ with level $\alpha$ is defined as:
        \begin{align*}
            ES_\alpha(Y) = \Ef[-Y | Y \leq - VaR_\alpha(Y)]
        \end{align*}
        where $VaR_\alpha(Y)$ is the \emph{Value-at-Risk} of the position $Y$.
    \end{definition}
    
    We refer the reader to \citet{artzner_1999} for further information on both of these risk measures. With the ES we look at the expectation of the losses suffered in the $\alpha \%$-worst cases. In contrast to the VaR, ES is indeed a coherent risk measure (it is sub-additive).

\section{Numerical Experiments}
\label{section_experiments}
In this section we investigate the performance of the Deep Hedging approach in comparison to a simple (dynamic) volume hedge. For illustration purposes we consider a Green PPA on wind onshore with delivery in just one hour. More complex examples with more realistic delivery periods will be subject of future work.

\subsection{Setup}
\label{subsection_parameters}
\paragraph{Model Parameters}

\begin{figure}
    \centering
    \includegraphics[width=0.32\textwidth]{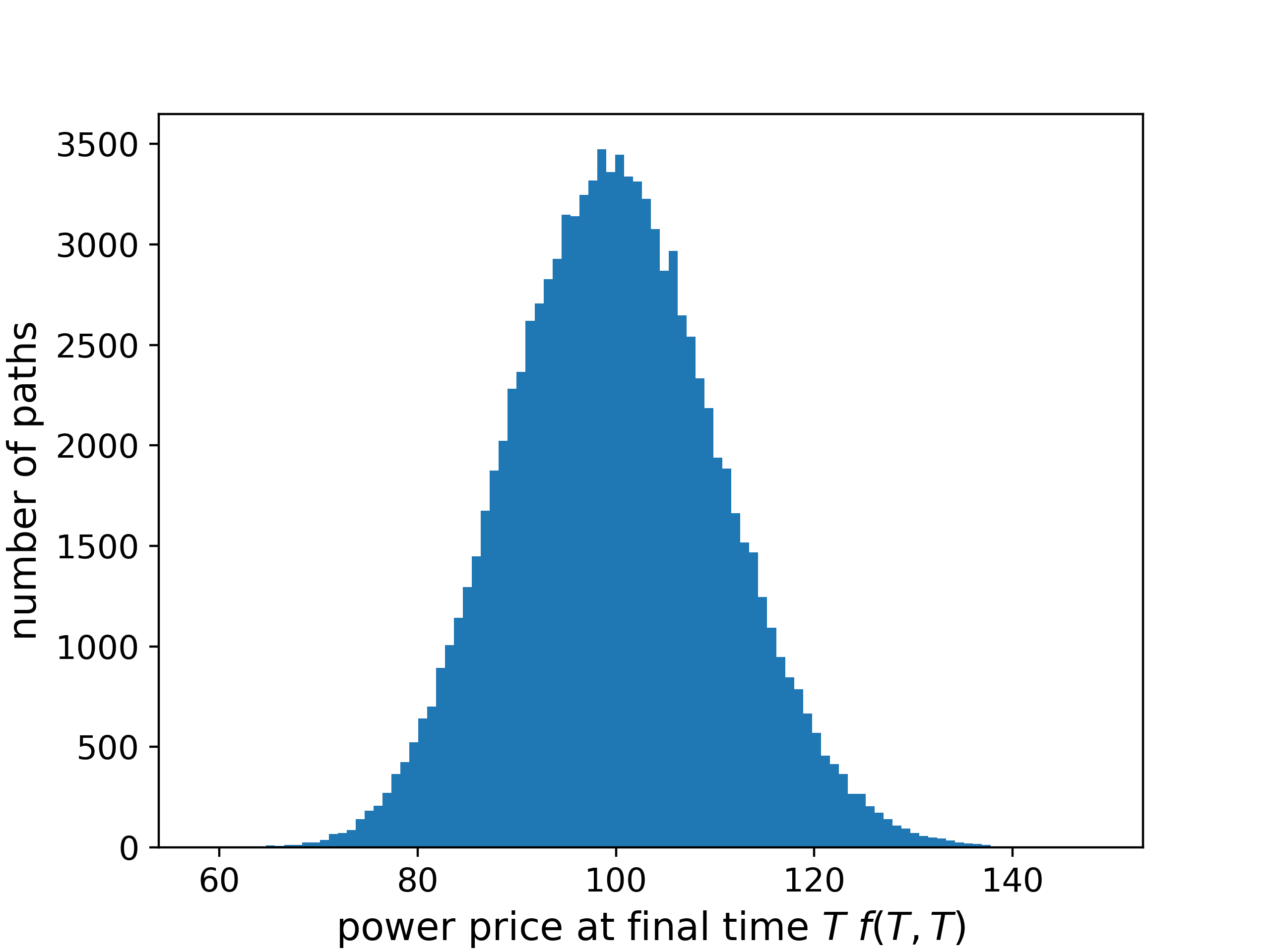}
    \includegraphics[width=0.32\textwidth]{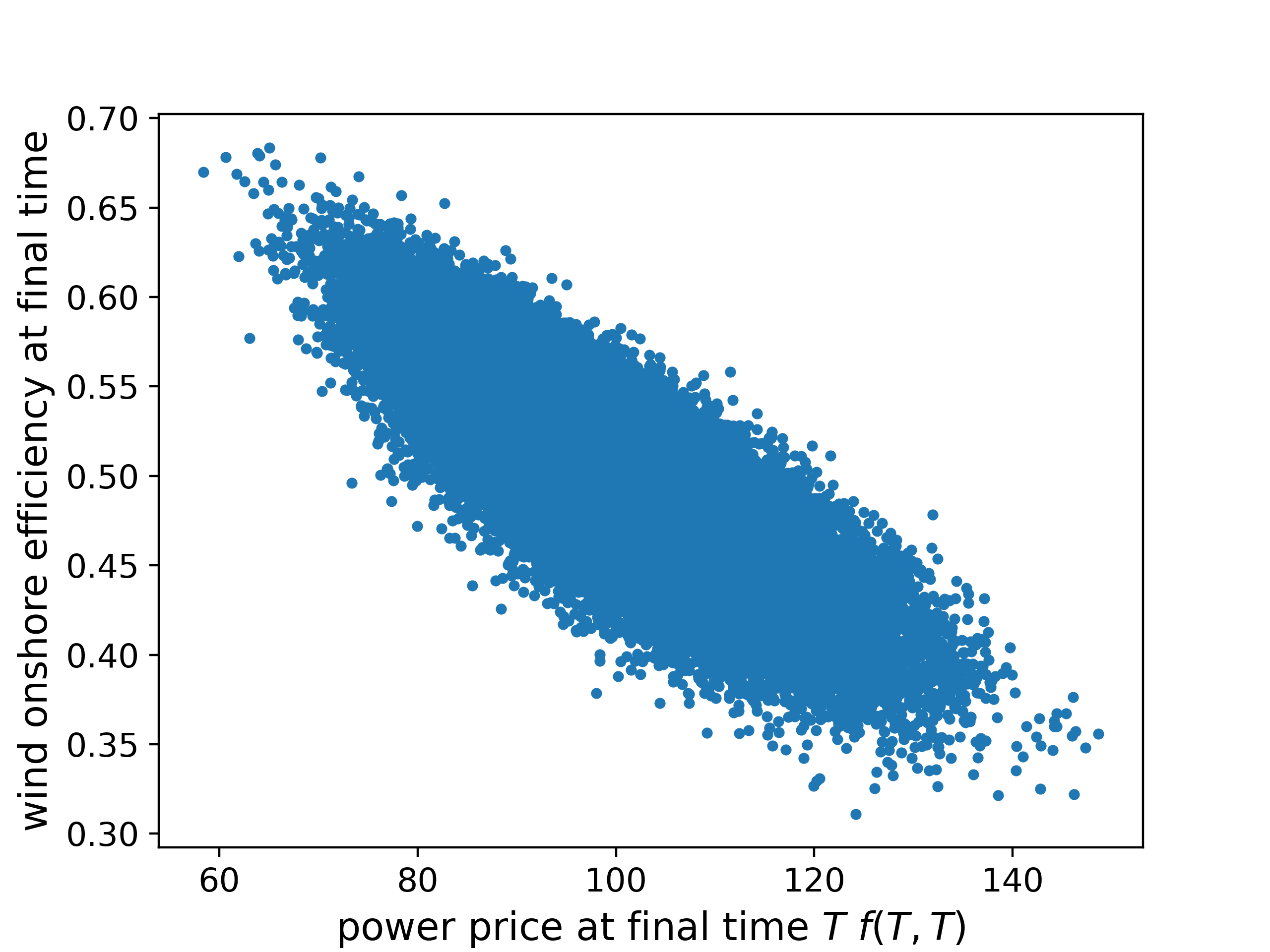}
    \includegraphics[width=0.32\textwidth]{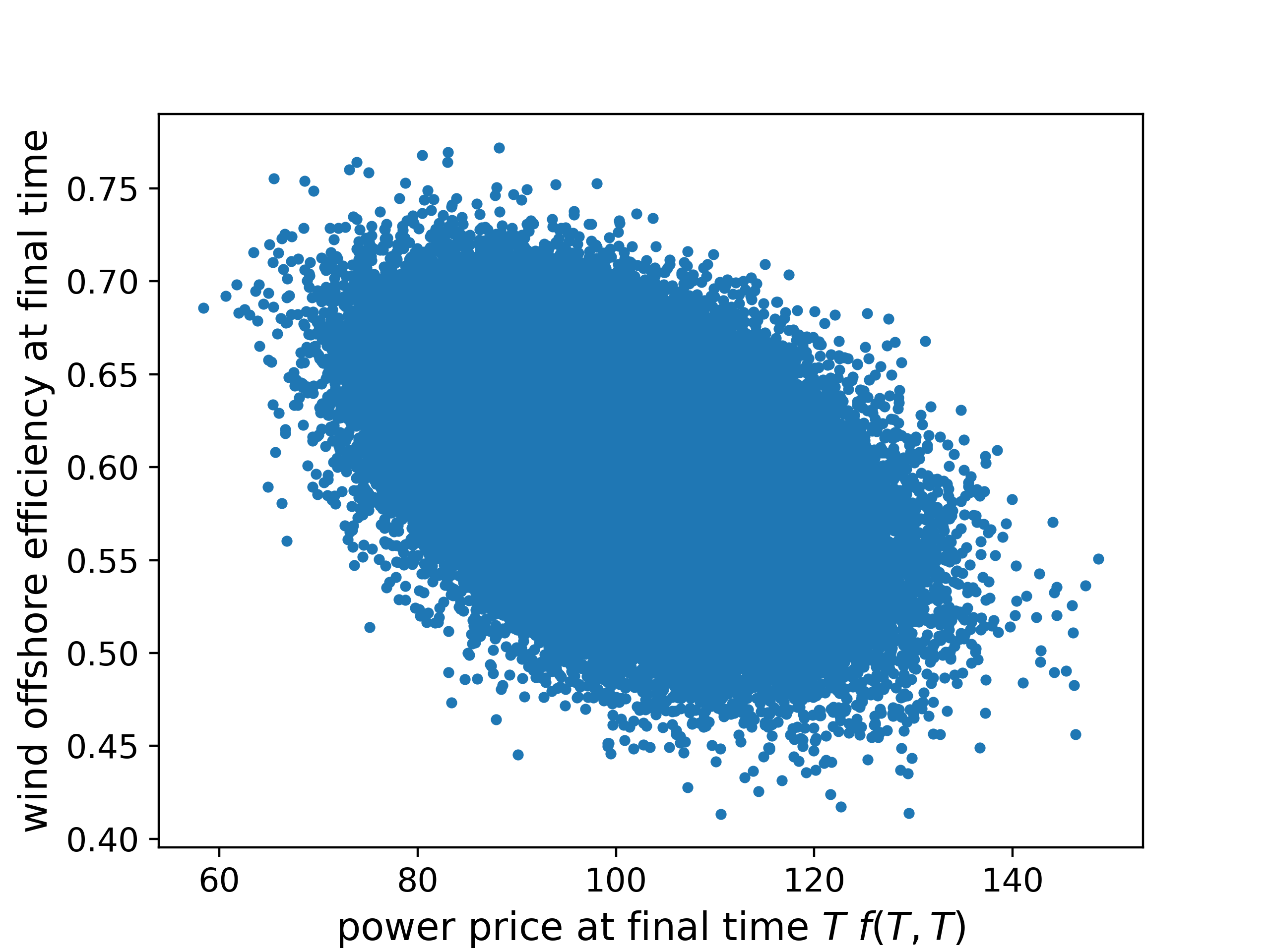}
    
    \includegraphics[width=0.32\textwidth]{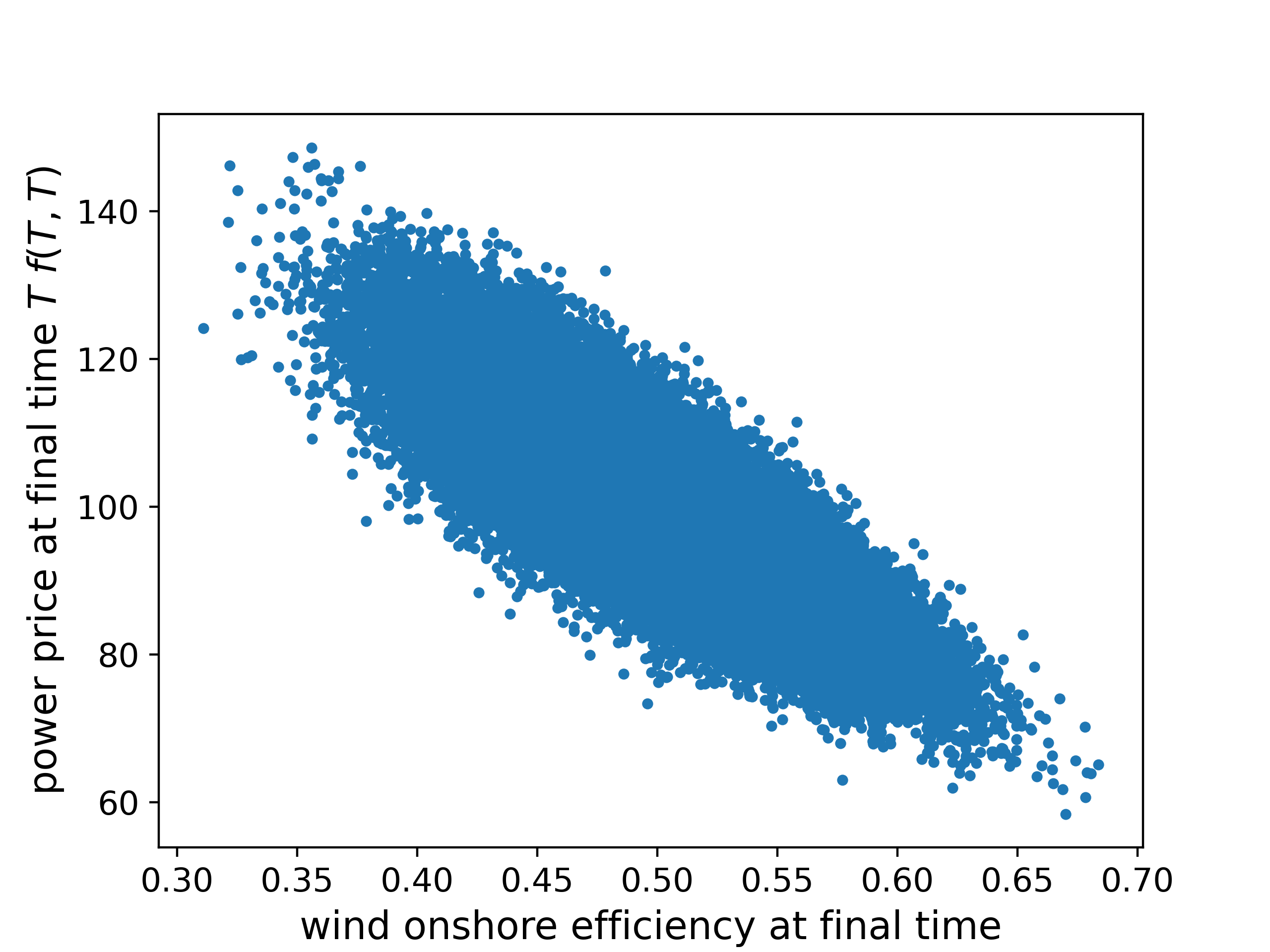}
    \includegraphics[width=0.32\textwidth]{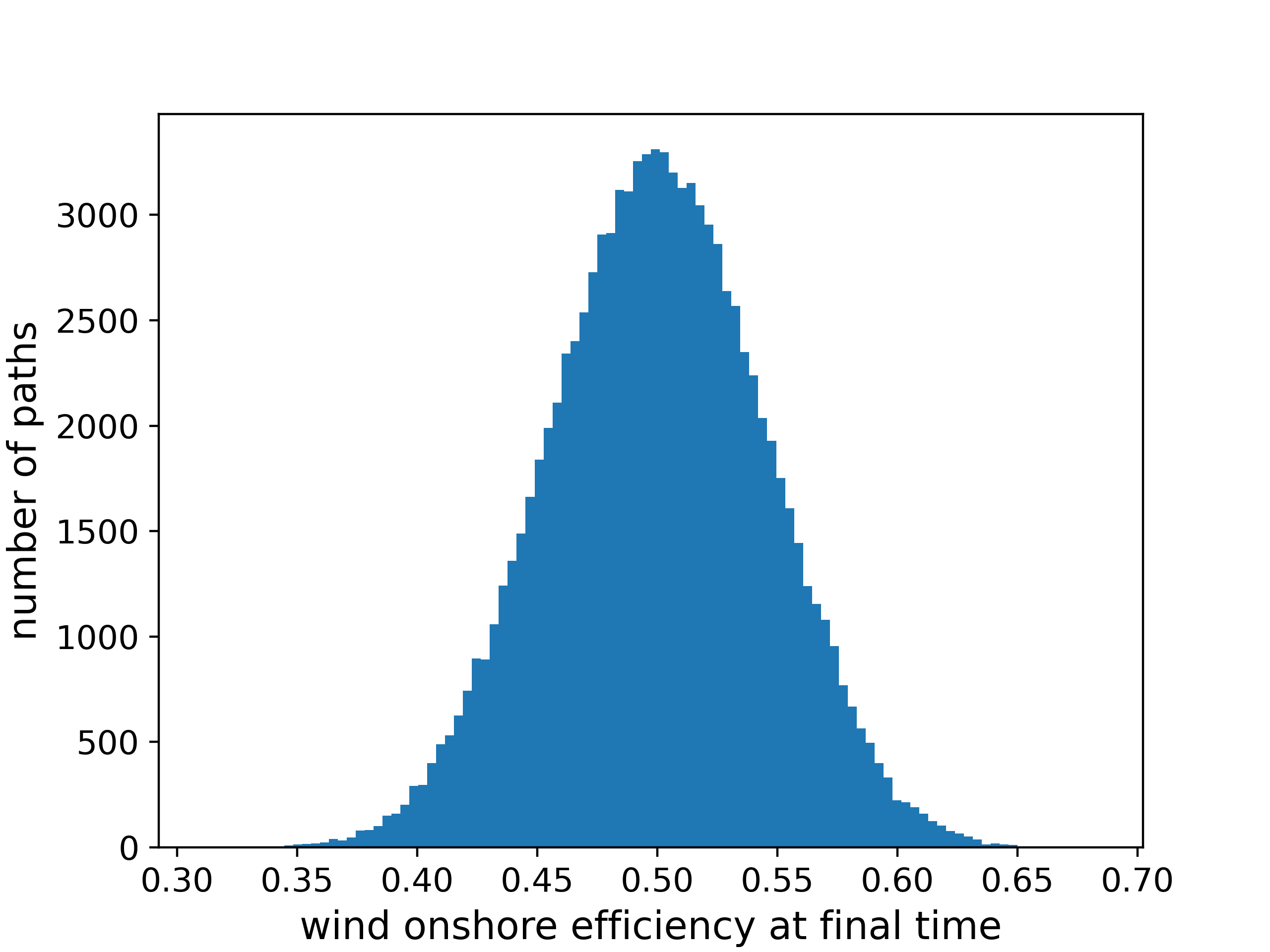}
    \includegraphics[width=0.32\textwidth]{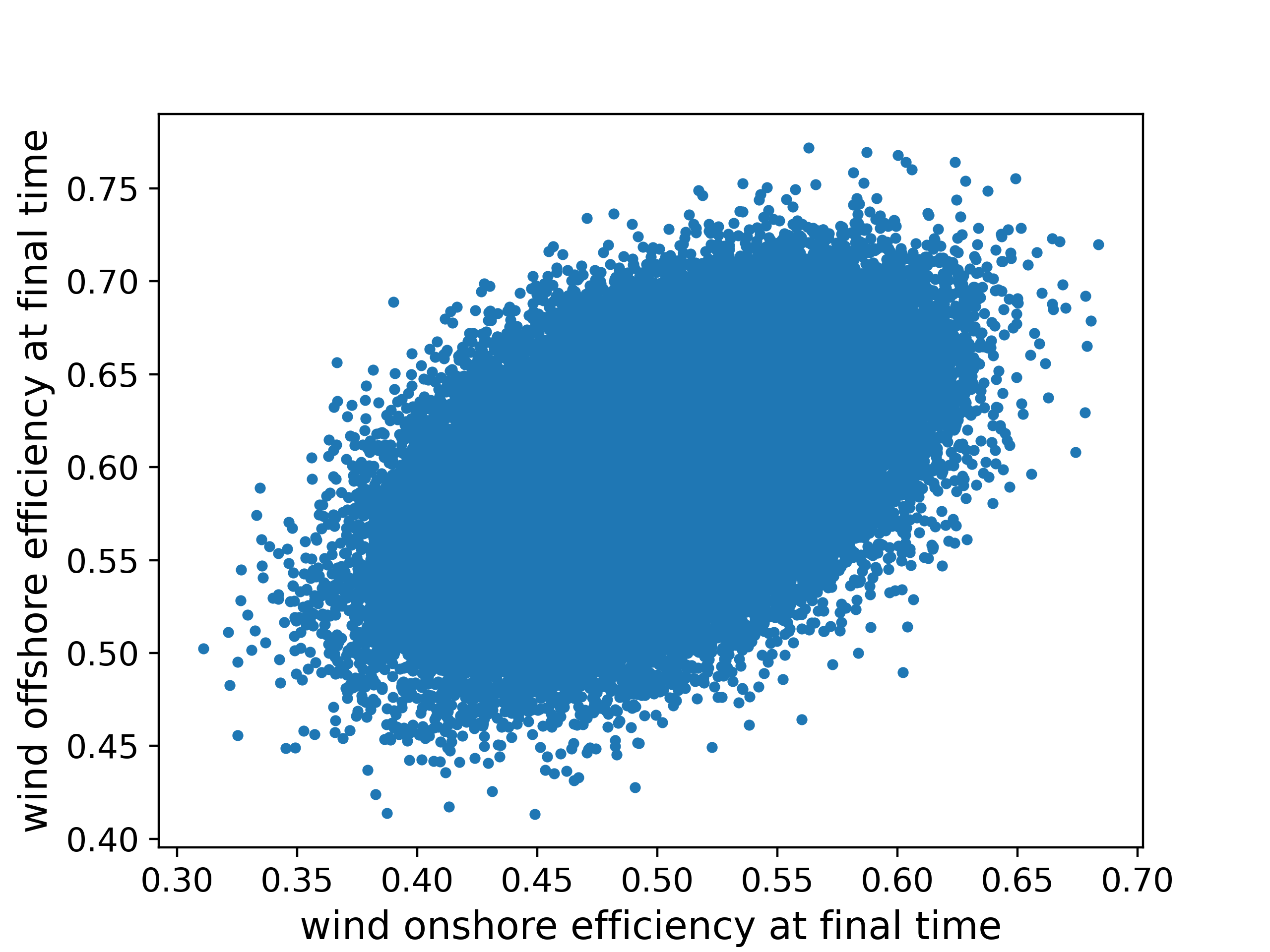}
    
    \includegraphics[width=0.32\textwidth]{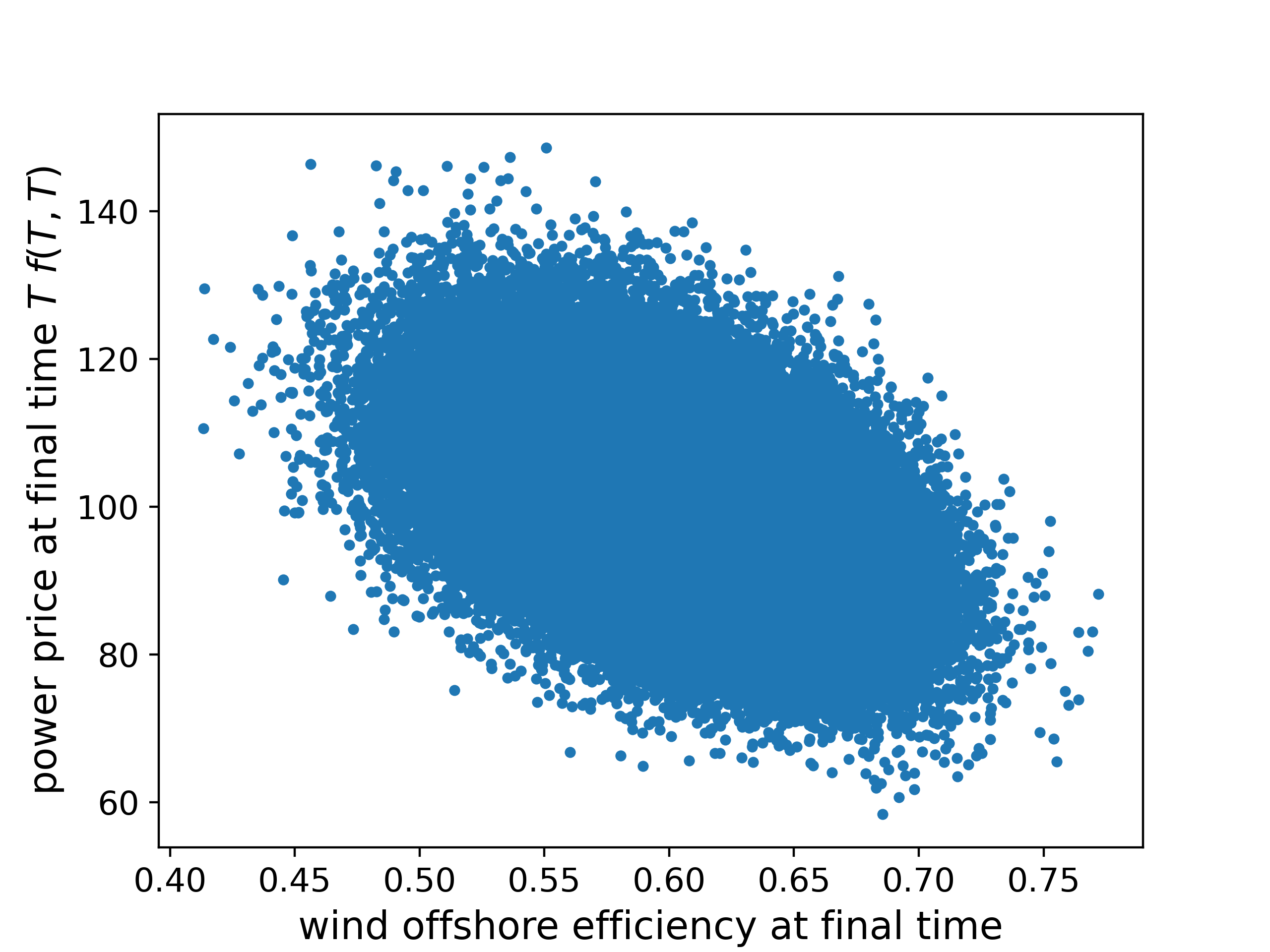}
    \includegraphics[width=0.32\textwidth]{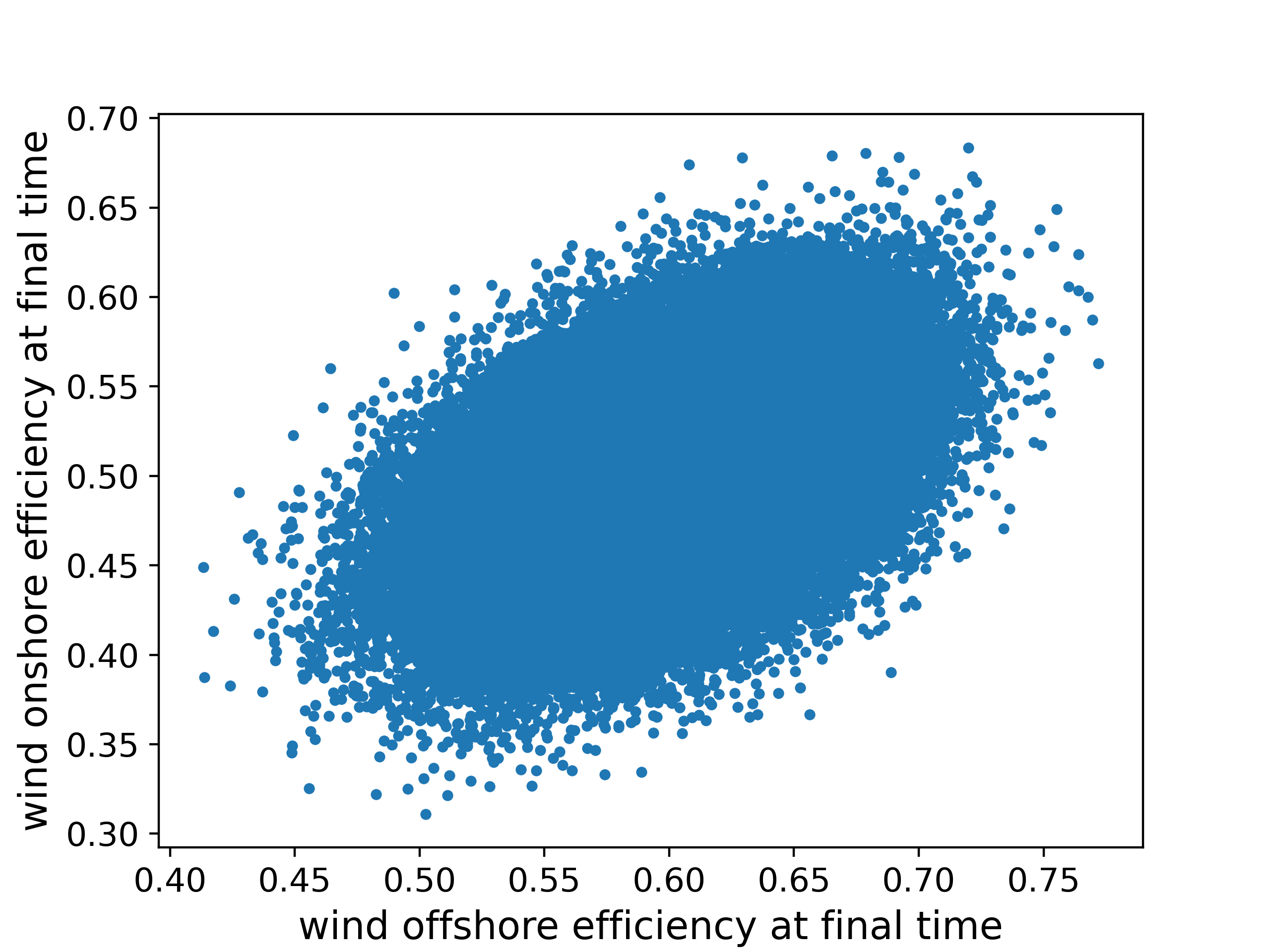}
    \includegraphics[width=0.32\textwidth]{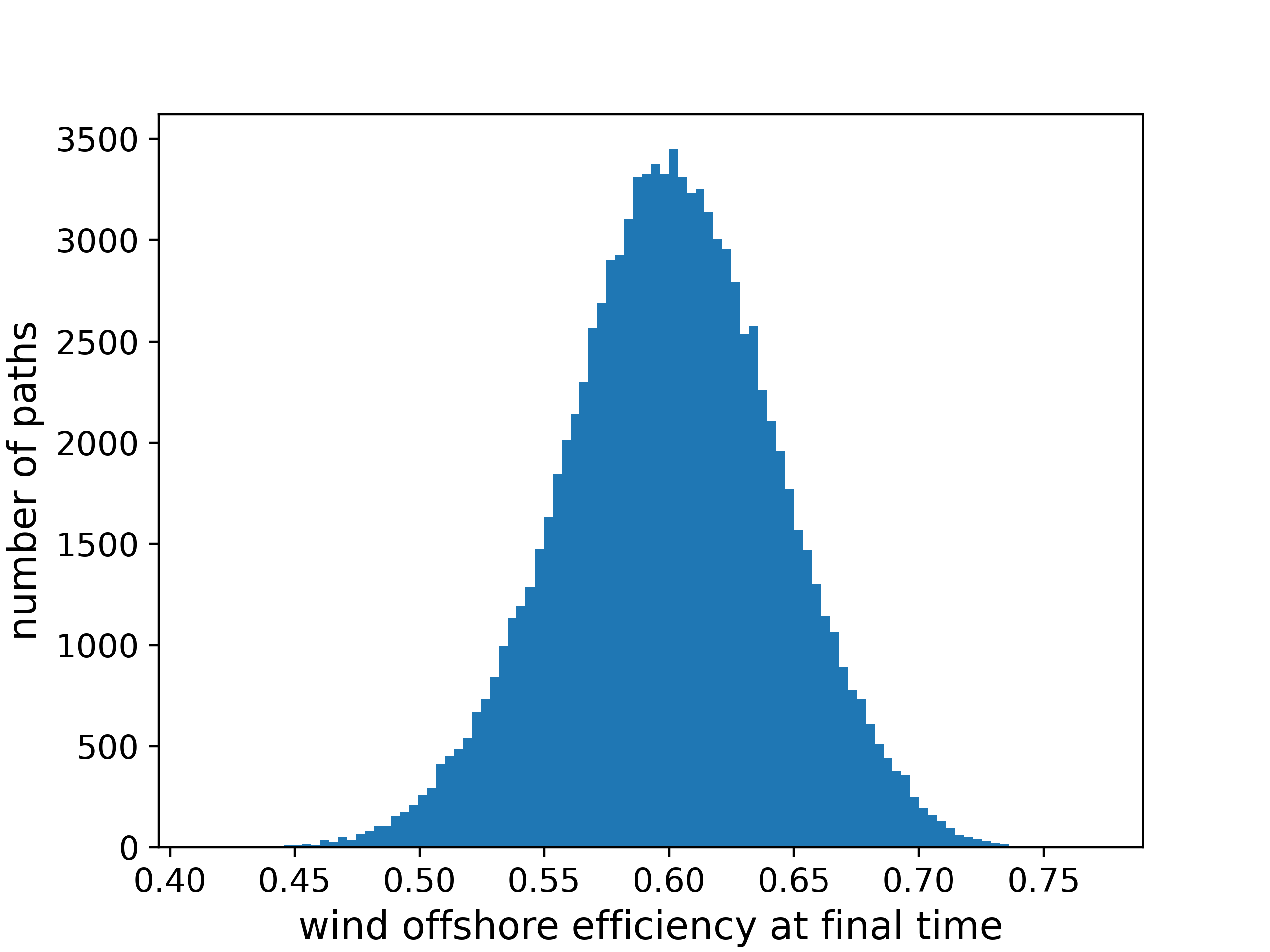}
    
    \caption{Histogram and scatter-plots for power price and wind efficiencies at the final time $T=48h$.}
    \label{fig:model}
\end{figure}
We consider a toy model for just two different renewables: Wind onshore, denoted by $Q_1$, and wind offshore, denoted $Q_2$. Solar radiation and also a finer resolution with respect to locations can easily be incorporated into the framework. If not stated otherwise the process parameters to model $p(t,T)$
are given by 
\begin{itemize}
    \item idiosyncratic component: $f(0,T)=100.0$, $\kappa^P=0.5$, $\sigma^P=0.8$,
    \item wind onshore: $Q_1(0,T)=0.5$, $\kappa_1=0.1$, $\sigma_1=3.0$, $w_1=0.91$,
    \item wind offshore: $Q_2(0,T)=0.6$, $\kappa_1=0.1$, $\sigma_1=3.0$, $w_2=0.09$,
\end{itemize}
and the correlation $\langle dW_1,dW_2\rangle=0.46$. As the sigmoid function in  \definitionref{def_quantities} used to obtain the efficiencies $Q_i$ we choose the logistic function
\[
    \varsigma(x) := \frac{1}{1+e^x}.
\]
For the approximation of the expected value in equation (\ref{eq:ForwardRelation}) we are using a combination of call prices (as described in \subsectionref{subsection_quantities}) We use 20 strikes uniformly distributed between -5 and 5. 
\figureref{fig:model} provides some intuition to this. It shows the histogram of power prices and wind efficiencies as well as their scatter plot. As we see, the final price is widely distributed in a range between 60 and 140. 

We consider a simple case of a Green PPA with fixed price $K$ with a delivery according to $Q_1$ over a one-hour period, i.e. the payoff $V(T)$ at expiry $T$ is given by
\[
    V(T) = c ~ Q_1(T,T) (f(T,T) - K),
\]
where $c>0$ is the plant's maximal capacity (set to 1 by default) which is multiplied by the efficiency $Q_1(T,T)$. The expiry $T$ is 48 hours in the future, $T$=48h. If not stated otherwise, the fixed PPA strike is $K = 100$ which equals the current forward price and we therefore say that the contract is at-the-money (ATM). For ease of presentation we will omit the final time $T$ in the following. We also remark, that, in reality, the delivery period of our toy model, i.e. one single hour, is of course not a tradable hedge product. As mentioned before, real delivery periods will be a topic of future research.

\paragraph{State Space}
We include the current timepoint $t$, the forward price $p_t(T)$ and the weather forecasts $Q_1(t,T)$ and $Q_2(t,T)$ into the state space, i.e. the state space is a four-dimensional vector defined by
\[
    X_t := \left(t,f(t,T),Q_1(t,T), Q_2(t,T) \right).
\]
Note that one may also include the current quantity held in the hedge instruments which is, in our case, the forward contract for the respective hour. This might be particularly important if we consider trading costs. In that case, the cost for adjusting a hedge at a certain timepoint depends on the residual of the old amount of the hedge and the new hedge quantity and therefore may influence the strategy itself. However, in this work we assume the absence of trading costs and therefore do not include current hedge quantities. More complex setups will be subject of a following paper.

\paragraph{Network Architecture and Training}
We conducted a variety of experiments with different network architectures and hyperparameters together with different network initialisations. The method seems to be quite robust with respect to these parameters and, if not stated otherwise, we are using a simple three-layer feedforward neural network with 64 neurons per layer and SELU (\emph{Scaled Exponential Linear Unit}) activation functions. We use the ADAM (\emph{Adaptive Moment Estimation}) optimiser with 800 epochs and batch size of 2000. We use a set of 100,000 paths for one training epoch. We apply an inverse time decay learning rate schedule, i.e.  learning rate $l(n)$ at the $n$-th optimiser step is  determined by
\[
l(n) = \frac{l_0}{1 + \alpha \cdot floor(\frac{n}{n_{step}})}
\]
with initial learning rate $l_0=2e-5$, $\alpha=0.2$, $n_{step}=4000$.

Discussion of network parameters and further methods to make the algorithm more robust will be the subject of a second paper.

\paragraph{Statistical Arbitrage}
Identifying hedging strategies with Machine Learning algorithms is prone to the phenomenon of \emph{statistical arbitrage}. This occurs if the algorithm "learns" to exploit technical properties of the models involved to create risk-free gains. We refer to the detailed discussion on this topic in \citet{buehler_2022}. We have applied our algorithm to a setup without the Green PPA, that is hedging an empty position. The resulting PnL was zero and the strategy learned was exactly the no-hedge strategy. This indicates that our model does not allow for statistical arbitrage.

\subsection{Results}
\label{subsection_results}

We investigate the hedge performance for two different risk measures: The expected shortfall (ES, see \definitionref{def_ES}) for the $5\%$ level and for the $30\%$ level. As benchmark strategies we also present results for the following three strategies: firstly, the trivial case of not hedging the PPA at all. Secondly, a static volume hedge which consists of selling the forecasted infeed quantity in the beginning. Thirdly, we benchmark against a dynamic volume hedge which consists of holding the quantity forecasted for the PPA at every time step. \tableref{tab:PnLResults} shows different statistics of the PnL (Profit and Loss) for the respective hedging strategies over 100,000 simulated paths, evaluated for a PPA starting at-the-money (forward price equals fixed price). We use the the parameters as specified in \subsectionref{subsection_parameters}.

\begin{table}[h]
    \centering
    \footnotesize
\begin{tabular}{lrrrrrr}
\toprule
{}                   &    Mean &    Variance &  p-Skewness &  1\% ES  &   5\% ES & 30\% ES \\
Hedge Strategy       &         &        &             &          &          &         \\
\midrule
No Hedge             & -0.3999 & 5.2794 &     -0.1609 & -15.7860 & -12.0629 & -6.6802 \\
Static Volume Hedge  & -0.3914 & 0.6039 &     -1.1283 &  -3.5152 &  -2.3222 & -1.0851 \\
Dynamic Volume Hedge & -0.3903 & 0.3979 &     -0.8189 &  -2.3414 &  -1.6025 & -0.8516 \\
30\% ES              & -0.3905 & 0.2934 &     -0.4860 &  -1.6320 &  -1.1889 & -0.7326 \\
5\% ES               & -0.3905 & 0.3013 &     -0.3993 &  -1.5778 &  -1.1759 & -0.7399 \\
\bottomrule
\end{tabular}
\caption{Statistics of the final PnL distribution for different hedging strategies computed over 100,000 paths.}
    \label{tab:PnLResults}
\end{table}

Firstly, we see that all hedging strategies have nearly the same PnL mean, although the hedging methods differ greatly in the other moments and statistics. As all strategies are self-financing the equivalent means come as no surprise. Different means for different strategies would constitute a statistical arbitrage opportunity as one could always go long in one strategy and short in another. Also, the mean for all strategies is slightly negative which stems from the cannibalisation effect. The rationale for this is as follows: If the forecasted quantities of the PPA are smaller today than they were yesterday, one would also likely experience larger prices. Then, one would have to rebuy previously hedged quantities at a higher price. The same logic applies for the opposite scenario of increased forecast quantities. Hence, both scenarios lead to additional costs for re-adjusting the hedge.

\begin{figure}
    \centering
    \includegraphics[width=0.9\textwidth]{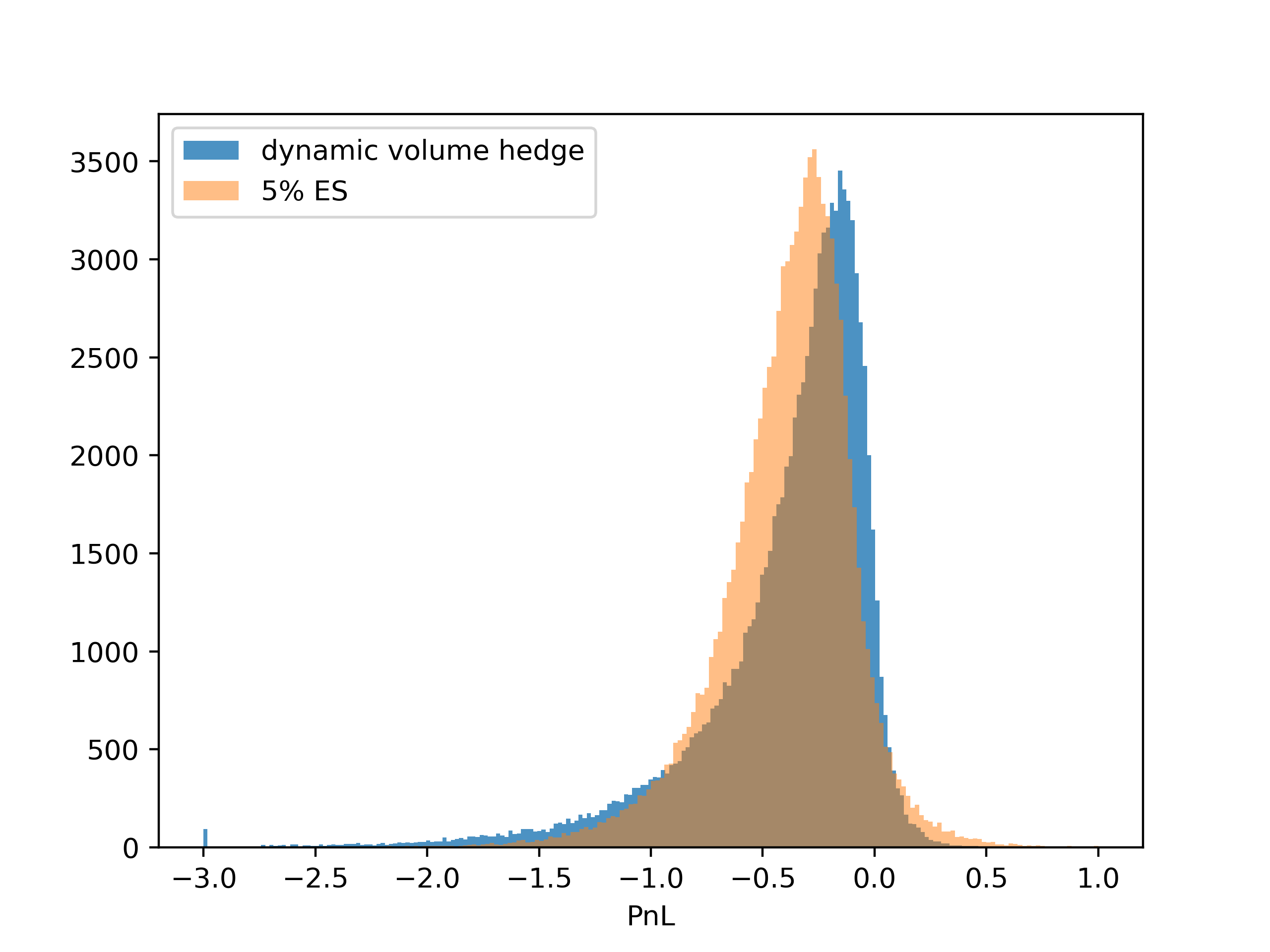} 
    \caption{Resulting PnL distributions for dynamic volume hedging and deep hedging using the 5\%-ES with 100.000 paths.}
    \label{fig:pnl_dist_utility_es}
\end{figure}

A good hedge should reduce variance. Thus, it is obvious that the case for which no hedging is done has by far the greatest variance when compared to all other hedges. The static volume hedge shows a massive decrease in variance but is still significantly worse than the more sophisticated strategies. The Deep Hedging approaches show the smallest variance and are also better than the dynamic volume hedge.

Negative skewness is another feature of the PnL-distribution that shows the cannibalisation effect (as described in the above example). The no-hedge case exhibits the lowest skewness (risks and chances almost symmetrical). With a hedge, while we offset price risks, we also have a two-way risk with respect to the (non-traded) quantity of the PPA. This results in larger negative skewness. Again, we see the Deep Hedge clearly outperforms the benchmarks.

\figureref{fig:pnl_dist_utility_es} shows the histogram of the PnL-distributions of the dynamic volume hedge and the 5\%-ES strategy, again for 100,000 simulated paths. One can see all of the above arguments about mean, variance and skewness. 

Coming back to \tableref{tab:PnLResults}, the difference between the dynamic volume hedge and Deep Hedging becomes larger if we consider the 5\% and 30\% expected shortfalls of the PnL-distribution. Note that the results of the Deep Hedging algorithm using the 5\% and 30\% ES as optimisation targets are very similar. Although the strategy calibrated to the 5\% ES provides slightly better results than the strategy calibrated for the 30\% ES for this risk measure (and vice versa), the difference is probably not significant. We remark that more complicated stochastic models would most likely increase the magnitude and significance of the difference between the two target risk measures. This is a topic of future research.

\begin{figure}[ht]
    \centering
    \includegraphics[width=0.85\textwidth]{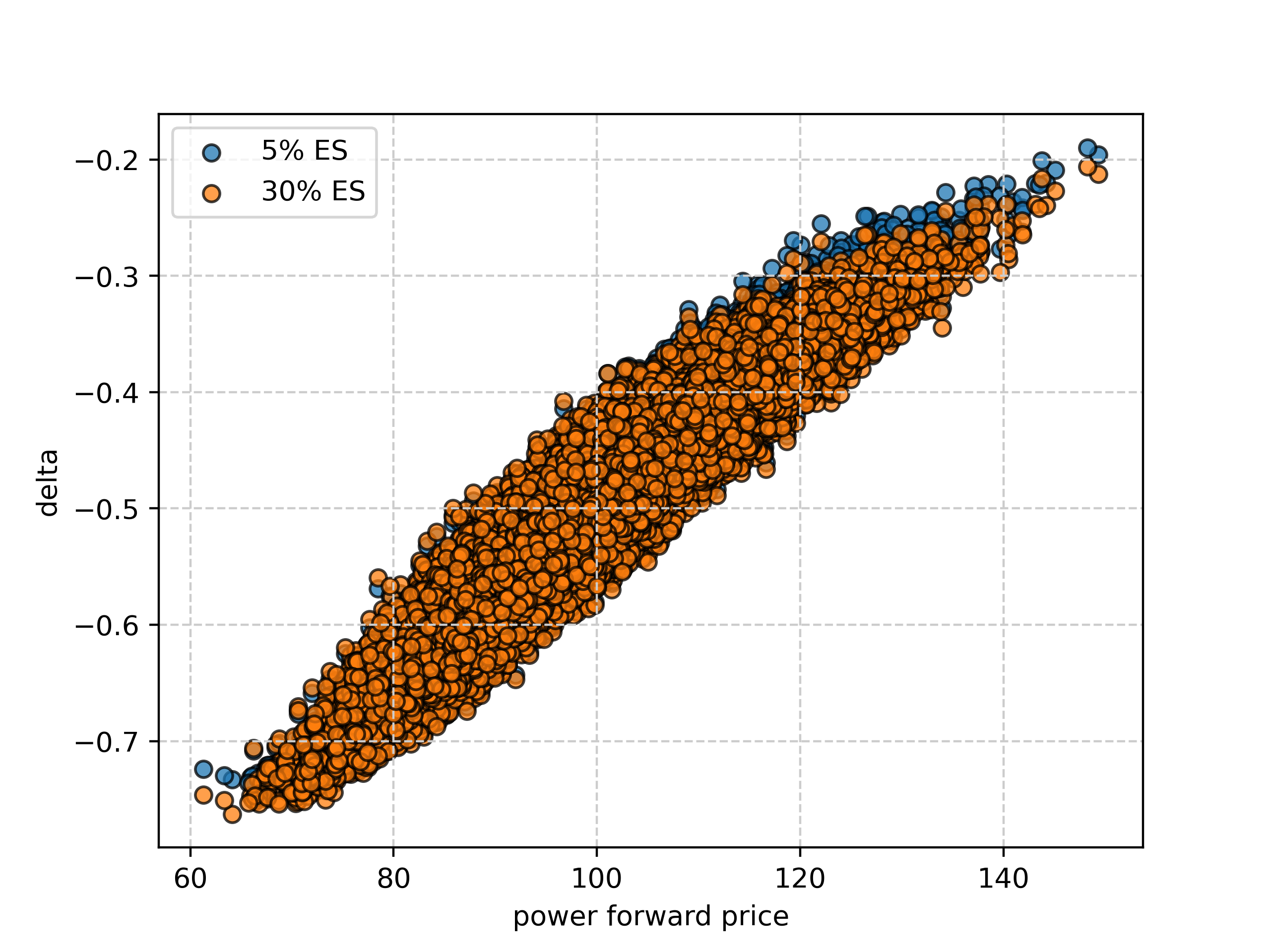}
    \caption{Case at-the-money: Delta versus price one hour before expiry on all test paths. This is the at-the-money case where PPA price equals the forward price of 100.}
    \label{fig:delta_vs_price_ATM}
\end{figure}

\begin{figure}[ht]
    \centering
    \includegraphics[width=0.85\textwidth]{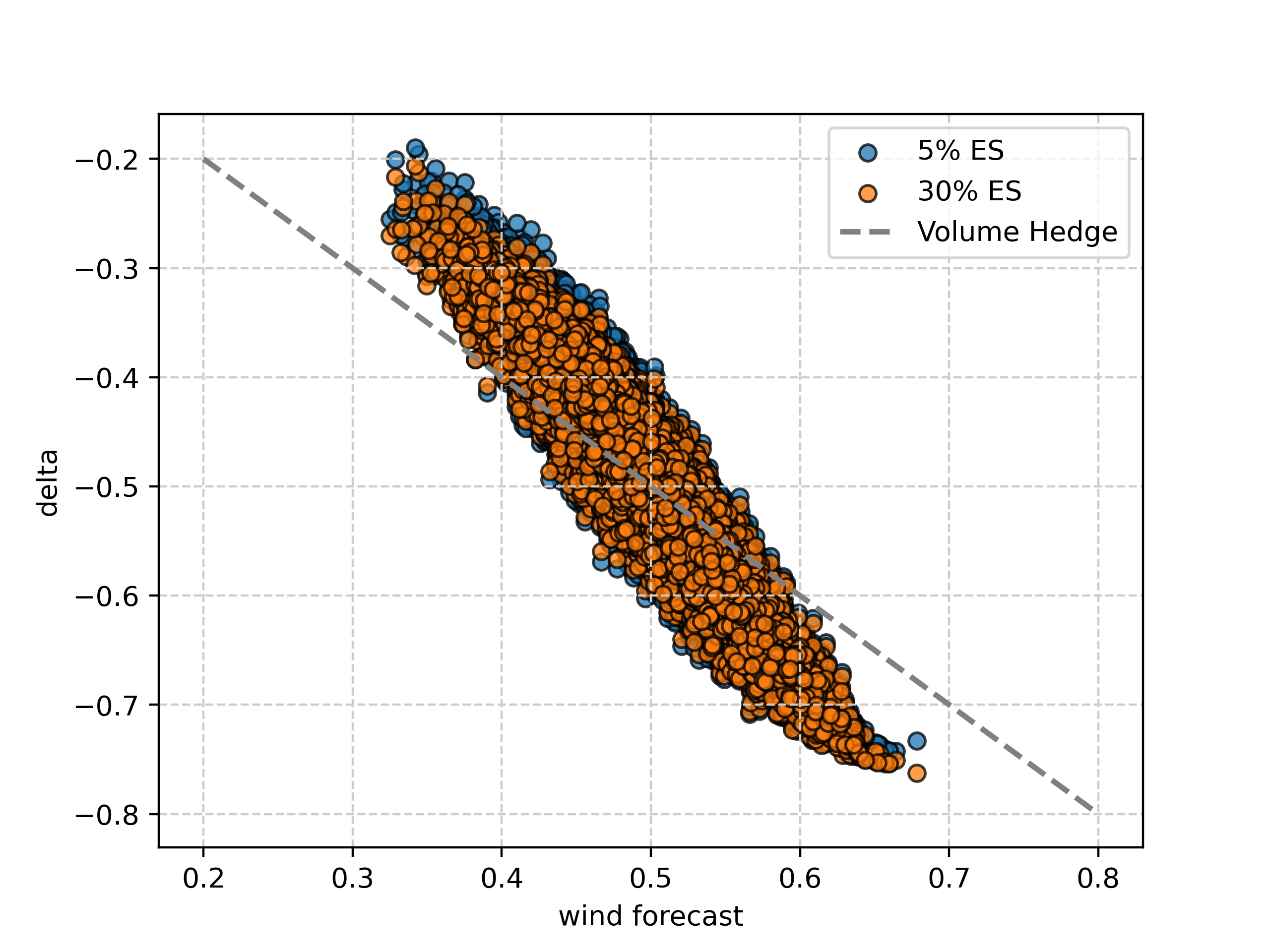}
    \caption{Case at-the-money: Delta wind forecast one hour before expiry on all test paths. This is the at-the-money case where PPA price equals the forward price of 100.}
    \label{fig:delta_vs_wind_ATM}
\end{figure}

So, how and when does the Deep Hedge achieve it's significant reduction in variance and skewness when compared to the benchmark volume hedge? In \figureref{fig:delta_vs_price_ATM} and \figureref{fig:delta_vs_wind_ATM} the delta (i.e. the hedging strategy $\delta$ as introduced in \sectionref{section_deepHedging}) is plotted against the power price as well as wind quantities. Here, we look at 100,000 paths and with one remaining hour until expiry of the PPA.

Both, the Deep Hedge minimising the 5\% (in orange) as well as that minimising the 30\% ES (in blue) suggest very similar hedge quantities. One can also see that the delta depends both on the current forward price (in \figureref{fig:delta_vs_price_ATM}) but also on the quantity forecast in (in \figureref{fig:delta_vs_wind_ATM}). This is, of course, in contrast to the volume hedge (denoted by the diagonal dashed line in \figureref{fig:delta_vs_wind_ATM}) for which just the forecasted quantity influences the hedging strategy. On average, we sell more power than the forecasted volume for those forecasts greater than 0.5 and less than the forecasted amount in the other cases. Here, the Deep Hedge tries to benefit from its learned knowledge of the cannibalisation effect and the relationship between prices and quantities in our model.

\begin{figure}[ht]
    \centering
    \includegraphics[width=0.49\textwidth]{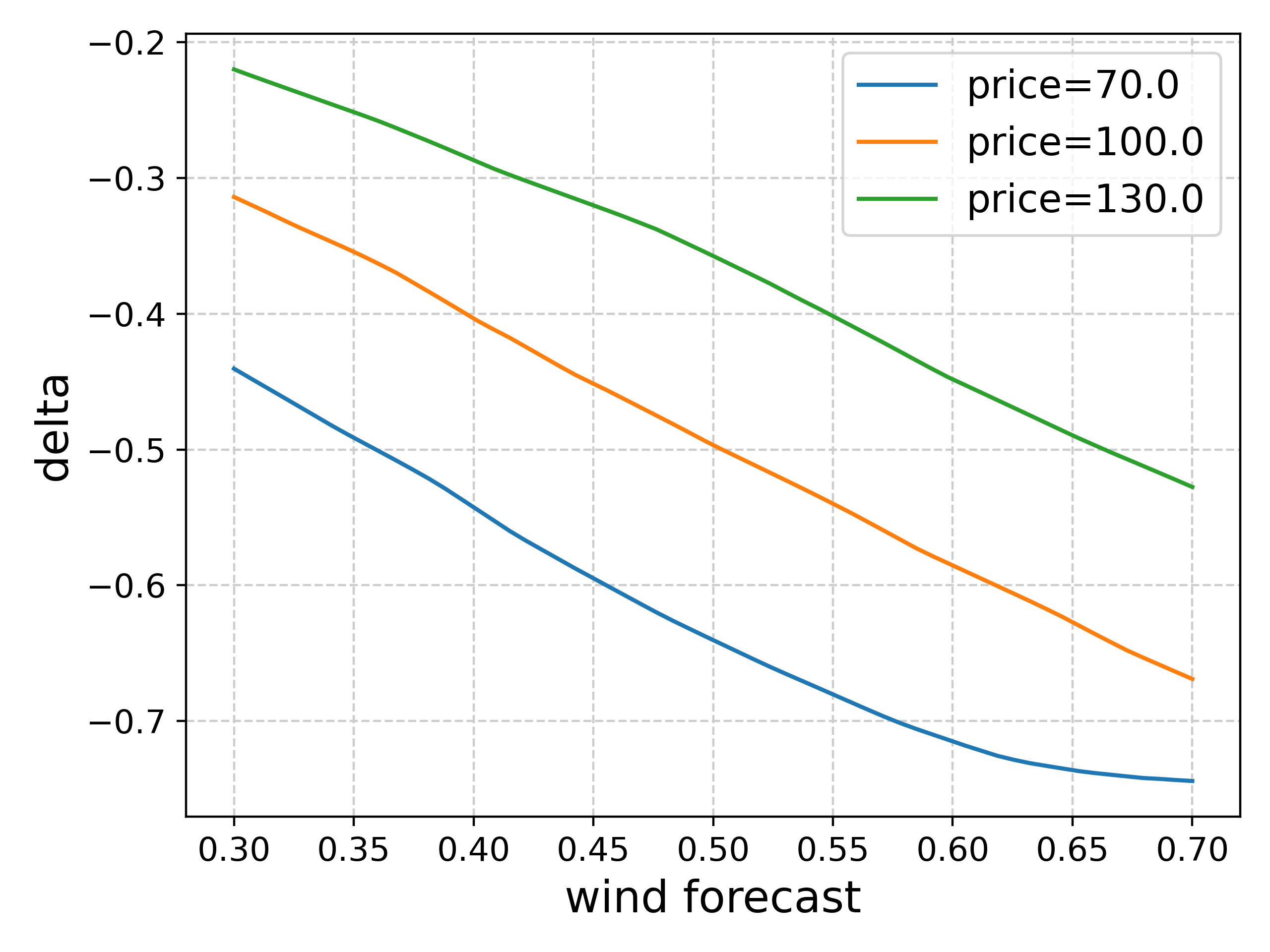}
    \includegraphics[width=0.49\textwidth]{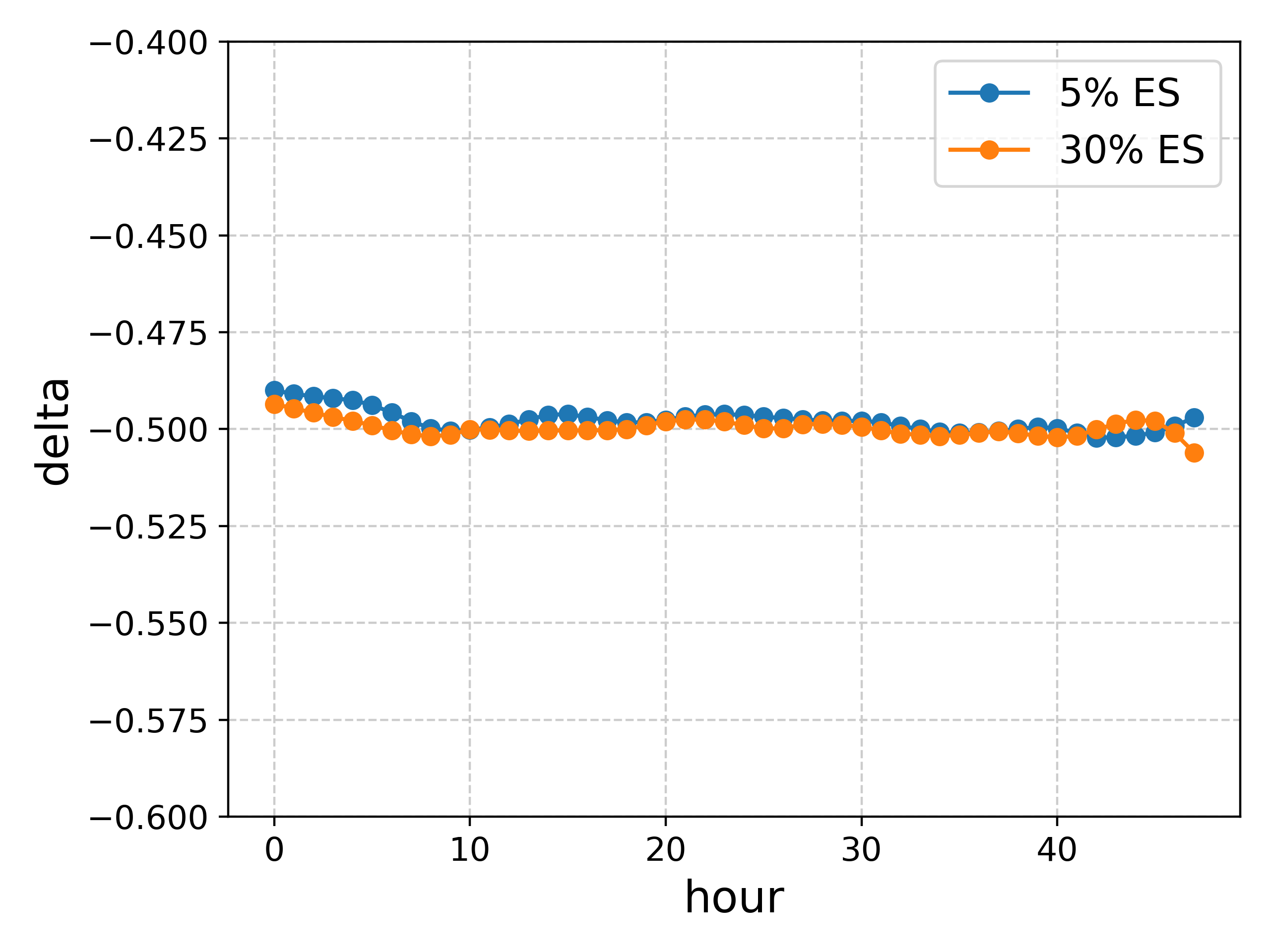}
    \caption{Delta for fixed prices and varying wind forecasts (for 5\% ES optimised strategy) (left) and mean of delta over all test paths for the hedging strategies trained to minimise 5\% and 30\% ES.  }
    \label{fig:mean_delta}
\end{figure}

This behaviour of the Deep Hedge is even more pronounced if the PPA has changed its moneyness, i.e. its fixed price is smaller than the forward price (in-the-money) or the other way around (out-of-the-money). The Deep Hedge then trades even more upside potential against a reduced downside: if the Green PPA is out-of-the-money it is better to sell more power than forecasted. This becomes clear in the left panel of \figureref{fig:mean_delta} where we plot the delta for different fixed forward price levels and varying wind forecasts. Nevertheless, as can be seen in the right panel of \figureref{fig:mean_delta}, the mean of the delta over all paths equals a delta very close to 0.5 (i.e. the initial forecast) at all timepoints and for both strategies.

\begin{figure}[ht]
    \centering
    \includegraphics[width=0.85\textwidth]{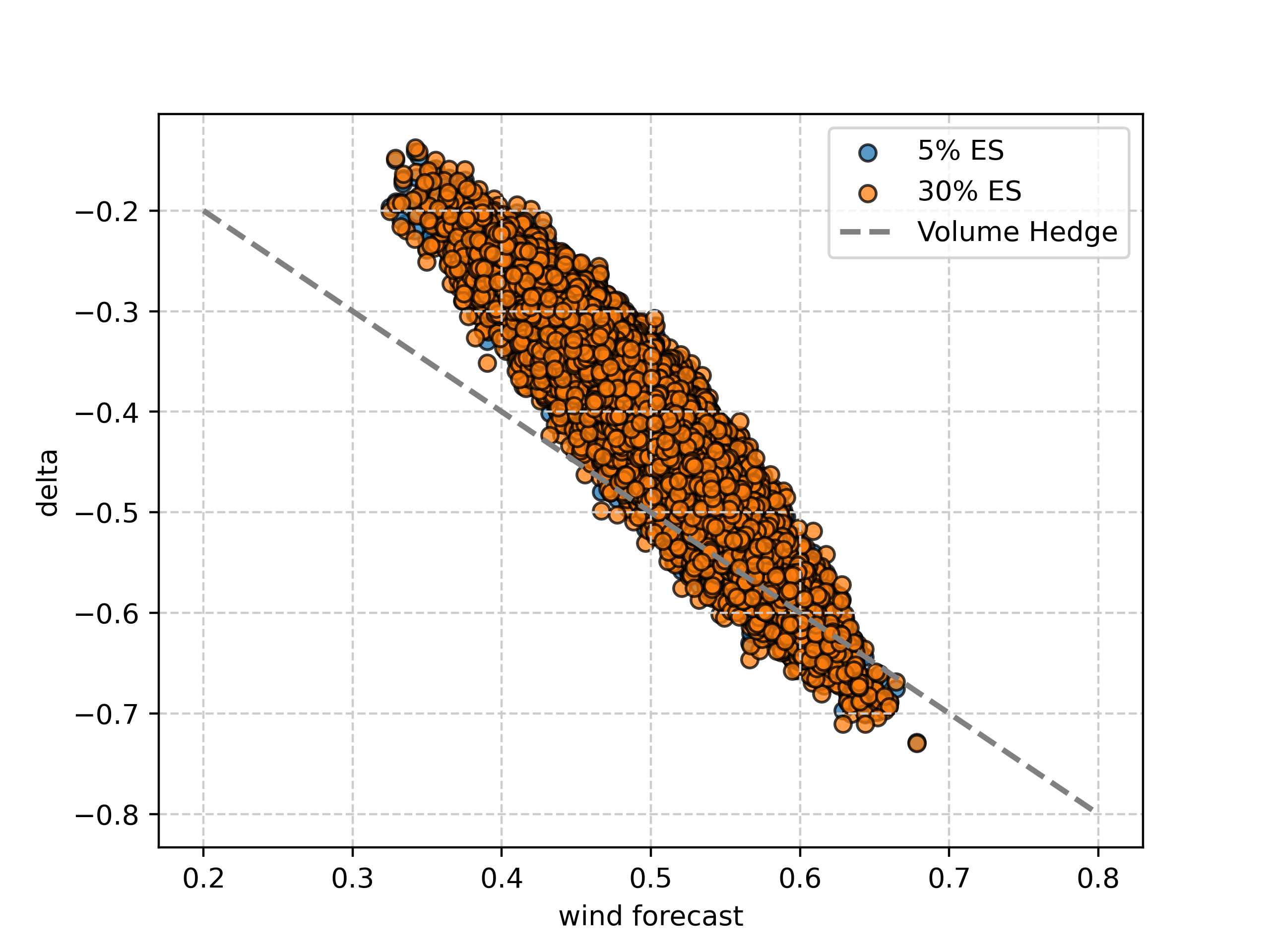}
    \caption{Case in-the-money: Delta and wind forecast one hour before expiry on all test paths. Here, the PPA price equals equal 80 compared to a forward price of 100.}
    \label{fig:delta_vs_price_wind_ITM}
\end{figure}

\figureref{fig:delta_vs_price_wind_ITM} and \figureref{fig:delta_vs_price_wind_OTM} again show the resulting hedge quantities for one remaining hour until delivery and 100,000 paths - but for the different cases of moneyness. \figureref{fig:delta_vs_price_wind_ITM} shows the in-the-money case and we can see how the Deep Hedge sells predominantly less quantity than the volume hedge. Thus, one does not lock-in gains on the price in order not to lose on the quantity risk. The opposite takes place in the out-of-the-money case as can be seen in \figureref{fig:delta_vs_price_wind_OTM}.

\begin{figure}[ht]
    \centering
    \includegraphics[width=0.85\textwidth]{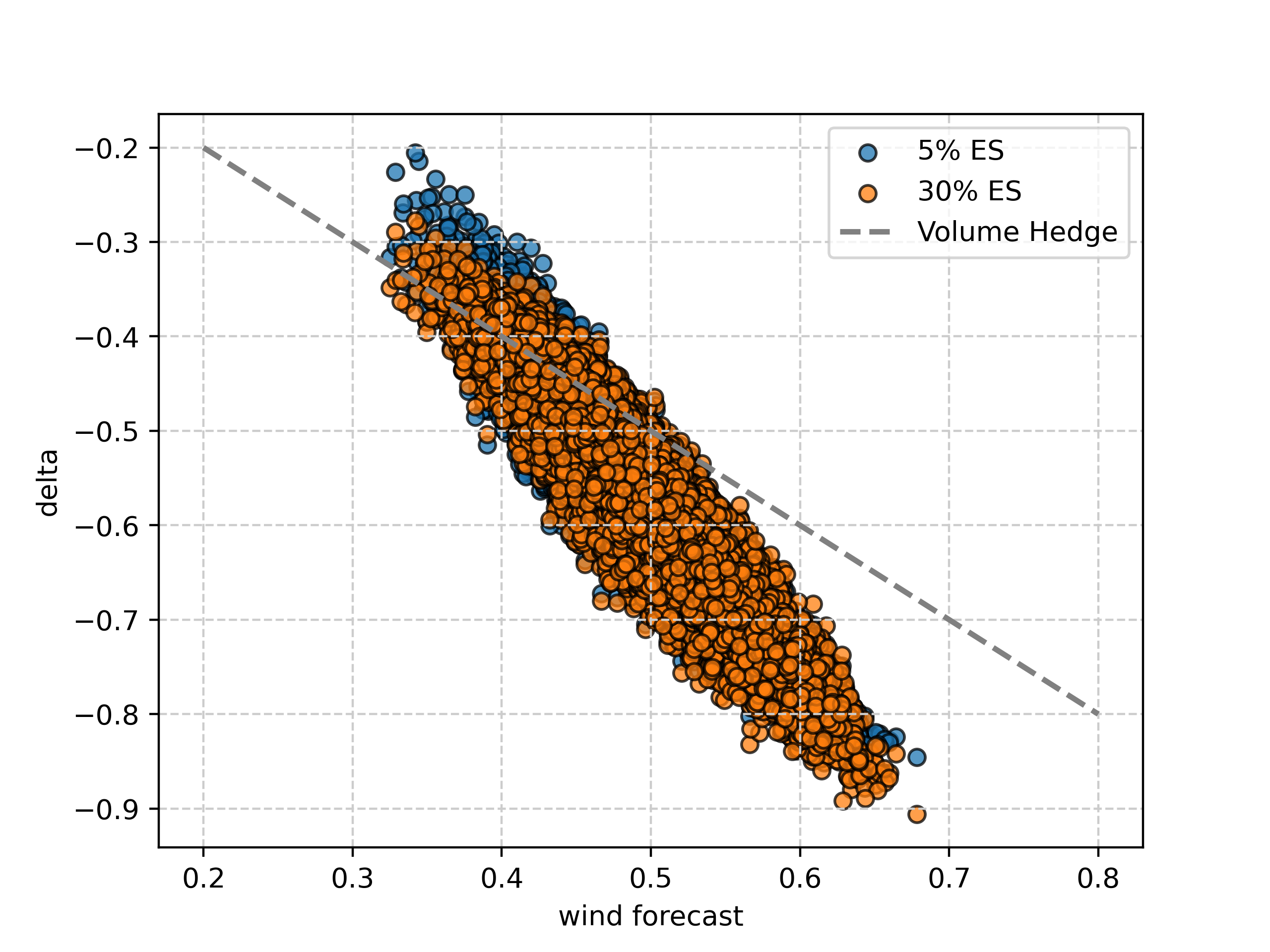}
    \caption{Case out-of-the-money: Delta versus wind forecast one hour before expiry on all test paths. Here the PPA price equals 120 compared to the forward price of 100.}
    \label{fig:delta_vs_price_wind_OTM}
\end{figure}

\section{Conclusion and Outlook}
\label{section_conclusion}
    In this study, we have presented a methodology for managing risks associated with positions in Green Power Purchase Agreements (PPAs). Green PPAs play an important role in the energy transition in liberalised energy markets. They are key in bringing renewable energy to the market and as such help in the fight against climate change.
    
    Energy traders usually take the side offering fixed prices for electricity in Green PPAs, thus securing steady cashflows for investors and stable prices for consumers of green electricity. In this role, they are highly exposed to market and weather risks as well as their negative interaction (the \emph{cannibalisation effect}). Thus, it is of the utmost importance for them to design good hedging strategies. The difficulty here is that trading Green PPAs takes place in a an incomplete market. While price risks can be mitigated by trading in the electricity market, this is not the case for the quantity risk. This depends on the weather (i.e. solar radiation or wind) which is not a traded entity. This challenge of finding effective hedging strategies has been the topic of our research.
    
    Our approach integrates a model for power forward prices with a stochastic framework for forecasting renewable energy infeeds. Specifically, infeed forecasts are modelled as the expected values derived from a logit-transformed Ornstein-Uhlenbeck process. Due to the absence of closed-form solutions, we approximate these expectations using a combination of call option prices.
    
    To optimise the trading strategy, we propose employing a neural network that minimises a predefined risk measure over the set of admissible strategies. This risk measure defines the individual risk appetite of the agent. While our analysis primarily focused on the expected shortfall as the risk criterion, the methodology is easily adapted to other risk measures, as well as utility functions.
    
    The results indicate that our approach yields superior risk-adjusted outcomes compared to a conventional dynamic volume hedge. The Machine Learning algorithms applied show stable training and prediction behaviour. Furthermore, making use of some intuition and \emph{domain knowledge} about electricity markets, the decisions suggested by the resulting hedge strategies can be \emph{interpreted} and make economical sense.
    
    Although transaction costs and bid-ask spreads were not explicitly incorporated in this study, their inclusion is straightforward within the proposed framework.
    
    Despite the promising results, our modelling approach remains relatively simplistic. Future research will further evaluate the adequacy of the proposed infeed and price models and explore potential refinements. This increase in complexity would, in turn, further justify the use of Machine Learning algorithms. Additionally, future research will extend the methodology to more complex settings: these include multi-asset portfolios, longer trading horizons as well as the incorporation of delivery periods.

\bibliography{References}

\begin{thebibliography}{14}
\providecommand{\natexlab}[1]{#1}
\providecommand{\url}[1]{\texttt{#1}}
\expandafter\ifx\csname urlstyle\endcsname\relax
  \providecommand{\doi}[1]{doi: #1}\else
  \providecommand{\doi}{doi: \begingroup \urlstyle{rm}\Url}\fi

\bibitem[Artzner et~al.(1999)Artzner, Delbaen, Eber, and Heath]{artzner_1999}
P.~Artzner, F.~Delbaen, J.-M. Eber, and D.~Heath.
\newblock Coherent {M}easures of {R}isk.
\newblock \emph{Mathematical Finance}, 9\penalty0 (3):\penalty0 203--228, 1999.
\newblock \doi{https://doi.org/10.1111/1467-9965.00068}.

\bibitem[Biegler-K\"onig and Pilz(2015)]{biegler_2015}
R.~Biegler-K\"onig and K.~F. Pilz.
\newblock Arbitrage-free {S}hifting of {P}rice {F}orward {C}urves - {A}n
  {E}fficient and {G}eneric {A}lgorithm.
\newblock \emph{Online paper}, 2015.
\newblock URL \url{https://ssrn.com/abstract=2706366}.

\bibitem[Biegler-K\"onig et~al.(2022)Biegler-K\"onig, Fischbach, and
  Pfingsten]{biegler_2022}
R.~Biegler-K\"onig, B.~Fischbach, and T.~Pfingsten.
\newblock Modelling {G}reen {PPA}s in {P}ower {M}arkets.
\newblock \emph{Online paper}, 2022.
\newblock URL \url{https://ssrn.com/abstract=4280606}.

\bibitem[Biegler-König(2020)]{biegler_2020}
R.~Biegler-König.
\newblock Hedging {S}trategies in {C}ommodity {M}arkets – {R}olling
  {I}ntrinsic and {D}elta {H}edging for {V}irtual {P}ower {P}lants.
\newblock \emph{Applied Mathematical Finance}, 27\penalty0 (6):\penalty0
  550--582, 2020.
\newblock \doi{10.1080/1350486X.2021.1898998}.
\newblock URL \url{https://doi.org/10.1080/1350486X.2021.1898998}.

\bibitem[Bingham and Kiesel(2004)]{kiesel_2004}
N.~Bingham and R.~Kiesel.
\newblock \emph{Risk-{N}eutral {V}aluation: Pricing and {H}edging of
  {F}inancial {D}erivatives}.
\newblock Springer Verlag, 2004.

\bibitem[Buehler et~al.(2019)Buehler, Gonon, Teichmann, and Wood]{buehler_2019}
H.~Buehler, L.~Gonon, J.~Teichmann, and B.~Wood.
\newblock Deep {H}edging.
\newblock \emph{Quantitative Finance}, 19\penalty0 (8):\penalty0 1271--1291,
  2019.
\newblock URL \url{https://doi.org/10.1080/14697688.2019.1571683}.

\bibitem[Buehler et~al.(2022)Buehler, Murray, Pakkanen, and Wood]{buehler_2022}
H.~Buehler, P.~Murray, M.~Pakkanen, and B.~Wood.
\newblock Deep {H}edging: {L}earning to {R}emove the {D}rift under {T}rading
  {F}rictions with {M}inimal {E}quivalent {N}ear-{M}artingale {M}easures, 2022.
\newblock URL \url{https://arxiv.org/abs/2111.07844}.

\bibitem[Garcke and Roscher(2023)]{garcke_2023}
J.~Garcke and R.~Roscher.
\newblock Explainable machine learning.
\newblock \emph{Machine {L}earning and {K}nowledge {E}xtraction}, 5\penalty0
  (1):\penalty0 169--170, 2023.
\newblock \doi{https://doi.org/10.3390/make5010010}.

\bibitem[Heath et~al.(1992)Heath, Jarrow, and Morton]{heath_1992}
D.~Heath, R.~Jarrow, and A.~Morton.
\newblock Bond pricing and the term structure of interest rates: a discrete
  time approximation.
\newblock \emph{Journal of {F}inancial and {Q}uantitative {A}nalysis},
  25:\penalty0 419--440, 1992.

\bibitem[Henderson and Hobson(2004)]{henderson_2004}
V.~Henderson and D.~Hobson.
\newblock Utility indifference pricing - an overview.
\newblock \emph{Volume on {I}ndifference {P}ricing}, 2004.

\bibitem[Hirth(2013)]{hirth_2013}
L.~Hirth.
\newblock The market value of variable renewables: The effect of solar wind
  power variability on their relative price.
\newblock \emph{Energy Economics}, 38:\penalty0 218--236, 2013.
\newblock \doi{https://doi.org/10.1016/j.eneco.2013.02.004}.

\bibitem[Tranberg et~al.(2020)Tranberg, Hansen, and Catania]{tranberg_2020}
B.~Tranberg, R.~Hansen, and L.~Catania.
\newblock Managing volumetric risk of long-term power purchase agreements.
\newblock \emph{Energy Economics}, 85:\penalty0 104--567, 2020.
\newblock \doi{https://doi.org/10.1016/j.eneco.2019.104567}.

\bibitem[Wagner(2012)]{wagner_2012}
A.~Wagner.
\newblock Residual demand modelling and applications to electricity pricing.
\newblock \emph{Berichte des Fraunhofer ITWM}, 213:\penalty0 1--28, 2012.

\bibitem[Wagner et~al.(2019)Wagner, Hinderks, and Korn]{wagner_2019}
A.~Wagner, W.~Hinderks, and R.~Korn.
\newblock A structural {H}eath–{J}arrow–{M}orton framework for consistent
  intraday spot and futures electricity prices.
\newblock \emph{Quantitative {F}inance}, 20\penalty0 (3), 2019.

\end{thebibliography}

\appendix

\section{Proof of \propositionref{prop_approx_call}}
\label{section_appendix_ouCall}

    We calculate the price of a call option on an underlying modelled by a OU-process with constant speed and vola and time-dependent level. Its dynamics and solution thus are:
    \begin{align*}
	    dX_t &= \kappa (\mu(t) - X_t) dt + \sigma dW_t \\
	    X_s &= X_t e^{-\kappa(s-t)} + \int^s_t e^{-\kappa(s-u)} \kappa \mu(u) du + \int^s_t e^{-\kappa(s-u)} \sigma dW_u
    \end{align*}

    The variance of the OU process is:
    \begin{align*}
    	\bar{\sigma}(t,T) = \Var(X_T | \mathcal{F}_t) = \sigma^2 \Ef\left[\int^T_t e^{-2\kappa(T-u)} du | \mathcal{F}_t\right] = \sigma^2 \frac{1}{2 \kappa} \left(1 - e^{-2\kappa(T-t)}\right)
    \end{align*}
    
    The call at time $t$ on $X_t$ with strike $K$ and expiry $T$ has fair value:
    \begin{align*}
    	C(t,T) = \Ef[(X_T - K)^+ | \mathcal{F}_t] &= 	\Ef[(X_T - K) \mathbbmss{1}_{X_T > K} | \mathcal{F}_t]
    \end{align*}

    The event of the above indicator function can be written as:
    \begin{align*}
    	X_T > K &\Rightarrow K <  X_t e^{-\kappa(T-t)} + \int^T_t e^{-\kappa(T-u)} \kappa \mu(u) du + \int^T_t e^{-\kappa(T-u)} \sigma dW_u \\
    	&\Rightarrow \bar{\sigma}(t,T) Z > K - X_t e^{-\kappa(T-t)} - \int^T_t e^{-\kappa(T-u)} \kappa \mu(u) du \\
    	&\Rightarrow Z > \frac{-\left(X_t e^{-\kappa(T-t)} + \int^T_t e^{-\kappa(T-u)} \kappa \mu(u) du  - K\right)}{\bar{\sigma}(t,T)} = -d
    \end{align*}
    where $Z \sim \mathcal{N}(0,1)$ and we used Ito's isometry. Going back to the payoff function, we now split the expectation into two parts:
    \begin{align*}
    	C(t,T) &= \underbrace{\Ef[X_T \mathbbmss{1}_{X_T > K} | \mathcal{F}_t]}_A - \underbrace{\Ef[K \mathbbmss{1}_{X_T > K} | \mathcal{F}_t]}_B
    \end{align*}

    We calculate both parts in turn. The second part is straightforward:
    \begin{align*}
	    B = K ~ \Pf(Z \geq -d) = K \Phi(d)
    \end{align*}
    For the first part we introduce another helper variable 
    \begin{align*}
        g(t,T) = X_t e^{-\kappa(T-t)} + \int^T_t e^{-\kappa(T-u)} \kappa \mu(u) du
    \end{align*}
    We then calculate:
    \begin{align*}
    	A &= \Ef\left[\left(g(t,T) + \bar{\sigma} Z\right) \mathbbmss{1}_{X_T > K} | \mathcal{F}_t\right] \\
    	&= g(t,T) \Phi(d) + \bar{\sigma} ~ \Ef\left[Z \mathbbmss{1}_{Z > -d} | \mathcal{F}_t\right]\\
    	&= g(t,T) \Phi(d) + \bar{\sigma} ~  \int^\infty_{-d} x \phi(x) dx \\
    	&= g(t,T) \Phi(d) + \bar{\sigma} ~ \phi(-d) \\
    	&= g(t,T) \Phi(d) + \bar{\sigma} ~ \phi(d)
    \end{align*}

    Collecting terms, we arrive at:
    \begin{align*}
        C(t,T) &= (g(t,T) - K) ~ \Phi(d) + \bar{\sigma} ~ \phi(d)
    \end{align*}
    which is the desired result.

\end{document}